\documentclass[sigconf]{acmart}

\usepackage{balance}
\usepackage{textpos}

\def\BibTeX{{\rm B\kern-.05em{\sc i\kern-.025em b}\kern-.08emT\kern-.1667em\lower.7ex\hbox{E}\kern-.125emX}}

 \setcopyright{none} 
\renewcommand\footnotetextcopyrightpermission[1]{}

\usepackage[noend]{algorithmic}
\usepackage[algoruled,nofillcomment,algo2e]{algorithm2e}
\usepackage{amsmath}
\usepackage{amsfonts}
\usepackage{amssymb}
\usepackage{mdframed}
\usepackage{framed}

\usepackage{comment}
\usepackage{natbib}
\usepackage{atbegshi,picture}
\usepackage{listings}
\usepackage{titlesec}
\titleformat*{\subsection}{\normalsize\bfseries}

\usepackage{caption}
\usepackage{subcaption}
\usepackage{graphics}
\usepackage{lipsum}

\usepackage[textsize=footnotesize,color=green!40]{todonotes}

\mathchardef\mhyphen="2D  
\renewcommand{\vec}[1]{{\bf #1}}
\newcommand{\mat}[1]{{\bf #1}}
\newcommand{\la}{\langle}
\newcommand{\ra}{\rangle}

\newcommand{\ot}{\textsc{OT}}
\renewcommand{\cot}[1]{\textsc{COT}_{#1}}

\newenvironment{protocol}[1][htb]
  {
   \begin{algorithm2e}[#1]
  }{\end{algorithm2e}}
\newcommand{\share}[1]{[\![{#1}]\!]}
\newcommand{\xorshare}[1]{\la{#1}\ra}
\newcommand{\A}{\mathsf{P_1}}
\newcommand{\B}{\mathsf{P_2}}
\newcommand{\partyi}{\mathsf{P_i}}

\newcommand{\QOO}{\texttt{QUOTIENT[AMSgrad=Ours, Norm=Ours]}}
\newcommand{\QSO}{\texttt{QUOTIENT[AMSgrad=Std, Norm=Ours]}}
\newcommand{\QOW}{\texttt{QUOTIENT[AMSgrad=Ours, Norm=WAGE]}}
\newcommand{\QSSO}{\texttt{QUOTIENT[SGD=Std, Norm=Ours]}}

\usepackage{amsthm}

\usepackage{tikz}
\usetikzlibrary{calc,backgrounds}
\usetikzlibrary{arrows,calc,decorations.pathreplacing,positioning,shapes,shapes.geometric,shapes.callouts}
\usetikzlibrary{shapes,intersections,external}
\usepackage{pgfplots}

\definecolor{hublue} {RGB}{  0, 55,108}
\definecolor{hured}  {RGB}{138, 15, 20}
\definecolor{hugreen}{RGB}{  0, 87, 44}
\definecolor{husand} {RGB}{210,192,103}
\definecolor{hugraygreen}{RGB}{209,209,194}
\definecolor{hugrayblue} {RGB}{189,202,211}

\definecolor{distinct11}{HTML}{77aadd}
\definecolor{distinct12}{HTML}{77cccc}
\definecolor{distinct13}{HTML}{88ccaa}
\definecolor{distinct14}{HTML}{dddd77}
\definecolor{distinct15}{HTML}{ddaa77}
\definecolor{distinct16}{HTML}{dd7788}
\definecolor{distinct17}{HTML}{cc99bb}

\definecolor{distinct21}{HTML}{4477aa}
\definecolor{distinct22}{HTML}{44aaaa}
\definecolor{distinct23}{HTML}{44aa77}
\definecolor{distinct24}{HTML}{aaaa44}
\definecolor{distinct25}{HTML}{aa7744}
\definecolor{distinct26}{HTML}{aa4455}
\definecolor{distinct27}{HTML}{aa4488}

\definecolor{distinct31}{HTML}{114477}
\definecolor{distinct32}{HTML}{117777}
\definecolor{distinct33}{HTML}{117744}
\definecolor{distinct34}{HTML}{777711}
\definecolor{distinct35}{HTML}{774411}
\definecolor{distinct36}{HTML}{771122}
\definecolor{distinct37}{HTML}{771155}

\newcommand{\relu}{\textsc{ReLU}}

\newcommand{\eb}{\mathbf{e}}
\newcommand{\Gb}{\mathbf{G}}
\newcommand{\ab}{\mathbf{a}}
\newcommand{\Wb}{\mathbf{W}}
\newcommand{\Mb}{\mathbf{M}}
\newcommand{\Vb}{\mathbf{V}}
\newcommand{\vb}{\mathbf{v}}

\usepackage{multirow}
\usepackage{tabularx, booktabs}
\newcolumntype{N}{>{\centering\arraybackslash}m{.5in}}
\usepackage{adjustbox}
\usepackage{enumitem}
\newcommand{\myparagraph}[1]{\vspace{1ex}\noindent{\bf #1}}

\tikzset{
  roundbox/.style={
    rectangle, draw, rounded corners
  },
  arrow/.style={
    ->, >=stealth, rounded corners
  },
  midcircle/.style={
    midway, circle, draw, fill=white
  },
  box p1/.style={
    roundbox,
    fill=white
  },
  box p2/.style={
    roundbox,
    fill=white
  },
  box mpc/.style={
    roundbox,
    fill=distinct12!66
  }
}

\begin{document}
\fancyhead{}

\title{QUOTIENT: Two-Party Secure Neural Network Training and Prediction}
  
\author{Nitin Agrawal}
\authornote{Both authors contributed equally to the paper. This work was partly done during an internship at the Alan Turing Institute.}
\email{nitin.agrawal@cs.ox.ac.uk}
\affiliation{
  \institution{University of Oxford}
  \streetaddress{}
  \city{}
  \country{}
}

\author{Ali Shahin Shamsabadi}
\email{a.shahinshamsabadi@qmul.ac.uk}
\affiliation{
  \institution{Queen Mary University of London}
  \streetaddress{}
  \city{}
  \country{}
}
\authornotemark[1]

\author{Matt J. Kusner}
\email{mkusner@turing.ac.uk}
\affiliation{
  \institution{University of Oxford}
  \streetaddress{}
  \city{}
  \country{}
}
\affiliation{
  \institution{The Alan Turing Institute}
  \streetaddress{}
  \city{}
  \country{}
}
\author{Adri\`{a} Gasc\'{o}n}
\email{agascon@turing.ac.uk}
\affiliation{
  \institution{University of Warwick}
  \streetaddress{}
  \city{}
  \country{}
}
\affiliation{
  \institution{The Alan Turing Institute}
  \streetaddress{}
  \city{}
  \country{}
}

\begin{abstract}

Recently, there has been a wealth of effort devoted to the design of secure protocols for machine learning tasks. Much of this is aimed at enabling secure \emph{prediction} from highly-accurate Deep Neural Networks (DNNs). However, as DNNs are trained on data, a key question is how such models can be also \emph{trained} securely. The few prior works on secure DNN training have focused either on designing custom protocols for existing training algorithms, or on developing tailored training algorithms and then applying generic secure protocols. In this work, we investigate the advantages of designing training algorithms alongside a novel secure protocol, incorporating optimizations on both fronts. We present QUOTIENT, a new method for discretized training of DNNs, along with a customized secure two-party protocol for it. QUOTIENT incorporates key components of state-of-the-art DNN training such as layer normalization and adaptive gradient methods, and improves upon the state-of-the-art in DNN training in two-party computation. Compared to prior work, we obtain an improvement of 50X in WAN time and 6\% in absolute accuracy.

\end{abstract}


\settopmatter{printfolios=true}

\maketitle

\section{Introduction}

The field of secure computation, and in particular Multi-Party Computation (MPC) techniques such as garbled circuits and lower level primitives like Oblivious Transfer (OT) have undergone very impressive developments in the last decade. This has been due to a sequence of engineering and theoretical breakthroughs, among which OT Extension~\cite{DBLP:conf/crypto/IshaiKNP03} is of special relevance.

However, classical generic secure computation protocols do not scale to real-world Machine Learning (ML) applications. To overcome this, 
recent works have combined different secure computation techniques to design custom protocols for specific ML tasks. This includes optimization of linear/logistic regressors and neural networks~\cite{nikolaenko_privacy-preserving_2013, petsregression, mohassel2017secureml,mohassel2018aby}, matrix factorization~\cite{nikolaenko_privacy-preserving_2013b}, constrained optimization~\cite{blindjustice}, and $k$-nearest neighbor classification~\cite{knn, knnwithgc}. For example, Nikolaenko et. al.~\cite{nikolaenko_privacy-preserving_2013} propose a protocol for secure distributed ridge regression that combines additive homomorphic encryption and garbled circuits, while previous works~\cite{petsregression,mohassel2017secureml} rely in part on OT for the same functionality. 

\myparagraph{Practical ML and Secure Computation: Two Ships Passing in the Night.}
While there have been massive practical developments in both cryptography and ML (including the works above), most recent works for model training~\cite{mohassel2017secureml, mohassel2018aby, wagh2019securenn, hesamifard2018privacy} and prediction~\cite{cryptopredict,cryptonets,tapas} on encrypted data are largely based on optimizations for {\em either} the ML model or the employed cryptographic techniques in isolation. In this work, we show that there is a benefit in taking a holistic approach to the problem. Specifically, our goal is to design an optimization algorithm alongside a secure computation protocol customized for it.

\myparagraph{Secure Distributed Deep Neural Network Training.} So far there has been little work on training Deep Neural Networks (DNNs) on encrypted data. The only works that we are aware of are ABY3~\cite{mohassel2018aby} and SecureML~\cite{mohassel2017secureml}. Different from this work, ABY3~\cite{mohassel2018aby} designs techniques for encrypted training of DNNs in the $3$-party case, and a majority of honest parties. The work most similar to ours is SecureML~\cite{mohassel2017secureml}. They propose techniques based on secret-sharing to implement a stochastic gradient descent procedure for training linear/logistic regressors and DNNs in two-party computation. While the presented techniques are practical and general, there are three notable downsides: 1. They require an ``offline'' phase, that while being data-independent, takes up most of the time (more than $80$ hours for a $3$-layer DNN on the MNIST dataset in the 2-Party Computation (2PC) setting); 2. 
Their techniques are not practical over WAN (more than $4277$ hours for a $3$-layer DNN on the MNIST dataset), restricting their protocols to the LAN setting; 3. The accuracy of the obtained models are lower than state-of-the-art deep learning methods. More recently, secure training using homomorphic encryption has been proposed~\cite{hesamifard2018privacy}. While the approach limits the communication overhead, it's estimated to require more than a year to train a $3$-layer network on the MNIST dataset, making it practically unrealizable.
Furthermore, the above works omit techniques necessary for modern DNNs training such as normalization \& adaptive gradient methods, instead relying on vanilla SGD with constant step size. 
In this paper we argue that significant changes are needed in the way ML models are trained in order for them to be suited for practical evaluation in MPC. Crucially, our results show that securely-trained DNNs do not need to take a big accuracy hit if secure protocols and ML models are customized \emph{jointly}.

\myparagraph{Our Contributions:}
In this work we present QUOTIENT, a new method for secure two-party training and evaluation of DNNs. Alongside we develop an implementation in secure computation with semi-honest security, which we call 2PC-QUOTIENT. Our main insight is that recent work on training deep networks in fixed-point for embedded devices~\cite{wu2018training} can be leveraged for secure training. Specifically, it contains useful primitives such as repeated quantization to low fixed-point precisions to stabilize optimization. However, out of the box this work does not lead to an efficient MPC protocol. To do so, we make the following contributions, both from the ML and the MPC perspectives:

1. We ternarize the network weights: $\Wb \in \{-1,0,1\}$ during the forward and backward passes. As a result, ternary matrix-vector multiplication becomes a crucial primitive for training. We then propose a specialized protocol for ternary matrix-vector multiplication based on Correlated Oblivious Transfer that combines Boolean-sharing and additive-sharing for efficiency. 

2. We further tailor the backward pass in an MPC-aware way, by replacing operations like quantization and normalization by alternatives with very efficient implementations in secure computation. We observe empirically in Section 6 that this change has no effect on accuracy. Alongside these changes, we extend the technique to residual layers~\cite{he2016deep}, a crucial building block for DNNs.

3. We design a new fixed-point optimization algorithm inspired by a state-of-the-art floating-point adaptive gradient optimization procedure~\cite{reddi2018convergence}. 

4. We implement and evaluate our proposal  in terms of accuracy and running time on a variety of real-world datasets. We achieve accuracy nearly matching state-of-the-art \emph{floating point accuracy} on 4 out of 5 datasets. Compared to state of the art 2PC secure DNN training \cite{mohassel2017secureml}, our techniques obtain $\sim$6\% absolute accuracy gains and $>\!50\times$ speedup over WAN for both training and prediction.

The rest of the paper is organized as follows. In the next section we introduce  notation and necessary background in ML and MPC. In Section~\ref{sec:ml} we present our proposed neural network primitives, and corresponding training procedure. In Section~\ref{section:crypto} we present a 2PC protocol for it. Finally, we present our experimental evaluation in Section~\ref{sec:experiments} and conclude with a short discussion.

\section{Overview and Problem Description}
Here, we introduce concepts in DNN training and inference, fixed-point encodings, and MPC applied to ML.
\subsection{Deep Neural Networks}

All DNNs are defined by a core set of operations, called a \emph{layer}. These layers are repeatedly applied to an input $\ab^0$ to produce a desired output $\ab^L$, where the number of repetitions (or layers) $L$ is called the \emph{depth} of the network. 
Every layer consists of a linear operation and a non-linear operation and, depending on the type of network, each layer takes on different forms. We describe three popular layer types that make up a large portion of state-of-the-art DNNs: \emph{fully-connected}, \emph{convolutional}, and \emph{residual}.
 
\begin{figure}[t]
\centering
\includegraphics[width=\columnwidth]{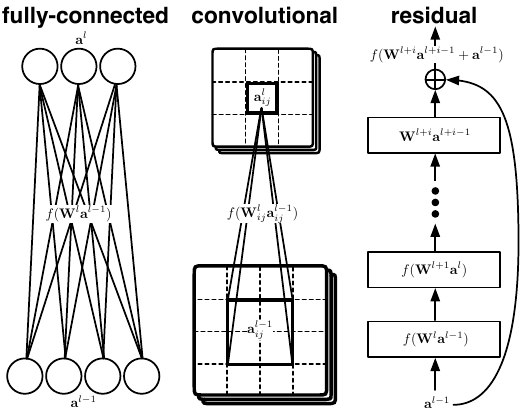}
\caption{We have protocols for three popular deep neural network layers: \emph{fully-connected}, \emph{convolutional}, and \emph{residual}.}
\label{fig:layers}
\end{figure}
\myparagraph{Fully-Connected Layers.}
To define a layer we simply need to define the linear and non-linear operations they use. In fully-connected layers these operations are as follows. 
Given an input $\ab^{l-1} \in \mathbb{R}^{h_{l-1}}$ of any fully-connected layer $l \in \{1,\ldots,L\}$, the layer performs two operations (1) a multiplication: $\Wb^{l} \ab^{l-1}$ with a weight matrix $\Wb^{l} \in \mathbb{R}^{h_{l} \times h_{l-1}}$; and (2) a non-linear operation, most commonly the Rectified Linear Unit: $\relu(x) = \max(x, 0)$. So the full layer computation is $\ab^{l} = \relu(\Wb^{l} \ab^{l-1})$, where $\ab^l$ is called the \emph{activation} of layer $l$. Note that `fully-connected' refers to the fact that any entry of $\ab^{l-1}$ is `connected' to the output $\ab^l$ via the weight matrix $\Wb^l$, shown schematically in Figure~\ref{fig:layers} (\emph{left}). 
Note that, more generic representations involve an intermediate step; adding a bias term $b^l$ to compute ($\Wb^{l} \ab^{l-1} + b^l$) before performing the non-linear operation ($\ab^{l} = \relu(\Wb^{l} \ab^{l-1}+b^l)$). In practice, this can be handled by suitably modifying $\Wb^l$ and $\ab^{l-1}$ prior to performing operation (1).

\myparagraph{Convolutional Layers.}
Convolutional layers are essentially fully-connected layers with a very particular connectivity structure. Whereas the input and output of a fully-connected layer are vectors, the input and output of a convolutional layer are 3rd-order tensors (i.e., an array of matrices). Specifically, the input $\ab^{l-1} \in \mathbb{R}^{h_{l-1} \times w_{l-1} \times c_{l-1}}$ can be thought of as an image with height $h_{l-1}$, width $w_{l-1}$, and channels $c_{l-1}$ (e.g., $c_{l-1}\!=\!3$ for RGB images). 

To map this to an output $\ab^l \in \mathbb{R}^{h_{l} \times w_{l} \times c_{l}}$, a convolutional layer repeatedly looks at small square regions of the input, and moves this region from left-to-right, from top-to-bottom, until the entire image has been passed over. Let $\ab^{l-1}_{ij} \in \mathbb{R}^{k_{l-1} \times k_{l-1} \times c_{l-1}}$ be the $k_{l-1} \times k_{l-1}$ region starting at entry $(i,j)$. This is then element-wise multiplied by a set of weights $\Wb^{l} \in \mathbb{R}^{k_l \times k_l \times c_{l-1} \times c_{l}}$ and summed across the first three dimensions of the tensor. This is followed by a non-linear operation, also (most often) the $\relu$ to produce an entry of the output $\ab^l_{ij} \in \mathbb{R}^{1 \times 1 \times 1 \times c_{l}}$. Note that $\ab^{l-1}_{ij}, \ab^{l}_{ij}, \Wb^l$ can all be vectorized such that $\ab^{l}_{ij} = f(\Wb^l \ab^{l-1}_{ij})$. This operation is shown in Figure~\ref{fig:layers} (\emph{center}). Apart from the dimensionality of the layers, other hyperparameters of these networks include: the size of square region, the number of `pixels' that square regions jump in successive iterations (called the \emph{stride}), and whether additional pixels with value $0$ are added around the image to adjust the output image size (called \emph{zero-padding}).

A common variant of a convolutional layer is a \emph{pooling layer.} In these cases, the weights are always fixed to $1$ and the non-linear function is simply $f(\cdot) \triangleq \max\{\cdot\}$ (other less popular $f(\cdot)$ include simple averaging or the Euclidean norm).

\myparagraph{Residual Layers.}
The final layer type we consider are residual layers~\cite{he2016deep}. Residual layers take an activation from a prior layer $\ab^{l-1}$ and add it to a later layer $l+i$ before the non-linear function $f(\cdot)$ (usually \relu) is applied: $f(\Wb^{l+i} \ab^{l+i-1} + \ab^{l-1})$. The intermediate layers are usually convolutional but may be any layer in principle. Figure~\ref{fig:layers} (\emph{right}) shows an example of the residual layer. DNNs with residual layers were the first to be successfully trained with more than $50$ layers, achieving state-of-the-art accuracy on many computer vision tasks, and are now building blocks in many DNNs.

\myparagraph{Training Deep Neural Networks.}
The goal of any machine learning classifier is to map an input $\ab^0$ to a desired output $y$ (often called a \emph{label}). To train DNNs to learn this mapping, we would like to adjust the DNNs weights $\{\Wb^l\}_{l=1}^L$ so that the output $\ab^L$ after $L$ layers is: $\ab^L = y$. 
To do so, the most popular method for training DNN weights is via Stochastic Gradient Descent (SGD). Specifically, given a training dataset of input-output pairs $\{(\ab^0_i, y_i)\}_{i=1}^n$, and a loss function $\ell(\ab^L, y)$ that measures the difference between prediction $\ab^L$ and output $y$ (e.g., the squared loss: $(\ab^L - y)^2$). SGD consists of the following sequence of steps:
  (1) The \emph{randomization step}: sample a random single input $\ab_i^0$, 
  (2) The \emph{forward pass}: pass $\ab_i^0$ through the network to produce prediction $\ab_i^L$,
  (3) The \emph{backward pass}: compute the gradients $\Gb^l$ and $\eb^l$ of the loss $\ell(\ab^L_i, y_i)$ with respect to each weight $\Wb^l$ and layer activation $\ab^l$ in the network, respectively: $\Gb^l = \frac{\partial \ell(\ab^L_i, y_i)}{\partial \Wb^l},   \eb^l = \frac{\partial \ell(\ab^L_i, y_i)}{\partial \ab^l}, \forall l\in[L]$, (4) The \emph{update step}: update each weight by this gradient: $\Wb_l = \Wb_l - \eta \Gb^l$, where $\eta$ is a constant called the \emph{learning rate}. 
These four steps, called an {\em iteration} are repeated for each input-output pair in the training set. A pass over the whole training set is called an \emph{epoch}. 
The first step is often generalized to sample a set of inputs, called a \emph{batch}, as this exploits the parallel nature of modern GPU hardware and has benefits in terms of convergence and generalization.

\subsubsection{State-Of-The-Art Training: Normalization \& Adaptive Step-Sizes.}
While the above ``vanilla'' gradient descent can produce reasonably accurate classifiers, alone they do not produce state-of-the-art results. Two critical changes are necessary: (1) normalization; and (2) adaptive step-sizes. We describe each of these in detail.

\myparagraph{Normalization.}
The first normalization technique introduced for modern DNNs was \emph{batch normalization}~\cite{ioffe2015batch}. It works by normalizing the activations $\ab^l$ of a given layer $l$ so that across a given batch of inputs $\ab^l$ has roughly zero mean and unit standard deviation. There is now a general consensus in the machine learning community that normalization is a key ingredient to accurate DNNs~\cite{bjorck2018understanding}. Since batch normalization, there have been a number of other successful normalization schemes including \emph{weight normalization}~\cite{salimans2016weight} and \emph{layer normalization}~\cite{ba2016layer}. The intuition behind why normalization helps is because it prevents the activations $\ab^l$ from growing too large to destabilize the optimization, especially for DNNs with many layers that have many nested multiplications. This means that one can increase the size of the learning rate $\eta$ (see the \emph{update step} above for how the learning rate is used in SGD), which speeds up optimization~\cite{bjorck2018understanding}. Normalization has been shown to yield speedups of an order of magnitude over non-normalized networks. Further, without normalization, not only is convergence slower, in some cases one cannot even reach lower minima~\cite{ba2016layer}. 

\myparagraph{Adaptive Step-Sizes.}
While normalization allows one to use larger learning rates $\eta$, it is still unclear how to choose the correct learning rate for efficient training: too small and the network takes an impractical amount of time to converge, too large and the training diverges. To address this 
there has been a very significant research effort into designing optimization procedures that adaptively adjust the learning rate during training \cite{duchi2011adaptive,kingma2014adam,reddi2018convergence}. So-called \emph{adaptive gradient methods} scale the learning rate by the magnitude of gradients found in previous iterations. This effectively creates a per-dimension step-size, adjusting it for every entry of the gradient $\Gb^l$ for all layers $l$. Without adaptive gradient methods, finding the right step-size for efficient convergence is extremely difficult. A prominent state-of-the-art adaptive gradient method is AMSgrad \cite{reddi2018convergence}. It was developed to address pitfalls with prior adaptive gradient methods. The authors demonstrate that it is able to converge on particularly difficult optimization problems that prior adaptive methods cannot. Even on datasets where prior adaptive methods converge, AMSgrad can cut the time to converge in half.

\subsubsection{Training and Inference with Fixed-Point Numbers.}
Motivated by the need to deploy DNNs on embedded and mobile devices with limited memory and power, significant research effort has been devoted to model quantization and compression. Often the goal is to rely solely on fixed-point encoding of real numbers. In fact, Tensorflow offers a lightweight variant to address this goal~\cite{tensorflowlite, tensorflowlitepaper}. 
These developments are useful for secure prediction. This is because cryptographic techniques scale with the circuit representation of the function being evaluated, and so a floating-point encoding and subsequent operations on that encoding are extremely costly. 
However, for the task of training, there are few works that perform all operations in fixed-point \cite{gupta2015deep,miyashita2016convolutional,koster2017flexpoint,wu2018training,hou2019analysis}. 
We start by reviewing fixed-point encodings, and the MPC techniques that we will consider. Then, in Section~\ref{section:ml}, we describe how this method can be modified and improved for secure training.

\myparagraph{Notation for Fixed-Point Encodings.}
We represent fixed-point numbers as a triple $x = (a, \ell, p)$, where $a\in \{-2^{\ell-1}-1,\ldots,2^{\ell-1}\}$ is the fixed-point integer, $\ell$ is its \emph{range} or \emph{bit-width}, and $p$ is its \emph{precision}. The rational number encoded by $x$ is $a/2^p$. For simplicity, we will often make $\ell$ implicit and write the rational $a/2^p$ as $a_{(p)}$. 
As our computations will be over $\ell$-bit values in 2's complement, overflows/underflows can happen, resulting in large big errors. This requires that our training procedure is very stable, within a very controlled range of potential values.

\subsection{MPC for Machine Learning}

In this section, we introduce MPC techniques, highlighting the trade-offs that inspire our design of secure training and prediction protocols. 
\label{subsec:mpc}
The goal of MPC protocols is to compute a public function $f(\cdot)$ on private data held by different parties. The computation is done in a way that reveals the final output of the computation to intended parties, and nothing else about the private inputs. MPC protocols work over a finite discrete domain, and thus the function $f(\cdot)$ must be defined accordingly. 
Generally, MPC protocols can be classified depending on
(i) a type of structure used to represent $f(\cdot)$ (generally either Boolean or integer-arithmetic circuits) 
and (ii) a scheme to secretly share values between parties, namely 
Boolean-sharing, additive-sharing,
Shamir-secret-sharing, and more complex variants (for more information see \cite{demmler_aby_2015}). 
The choices (i) and (ii) define the computational properties
of an MPC protocol. 

Additive-sharing protocols are very efficient for computations over large integral domains that do not involve comparisons. This includes sequences of basic linear algebra operations, such as matrix-vector multiplications. Specifically, additions are extremely cheap as they can be performed locally, while multiplications are more expensive. On the other hand, computations involving comparisons require computing a costly bit-decomposition of the values.

In contrast to additive-sharing, protocols based on Boolean-sharing are well-suited for computations easily represented as Boolean circuits, such as division/multiplication by powers of two (via bit shifting), comparisons, and $\texttt{sign}()$. They are slower at addition and multiplication which require adder and multiplier circuits.

These trade-offs lead to a natural idea recently exploited in several works in the secure computation for ML: one should design protocols that are customized to full algorithms, such as the training of linear/logistic regressors and DNNs~\cite{nikolaenko_privacy-preserving_2013, petsregression, mohassel2017secureml}, matrix factorization~\cite{nikolaenko_privacy-preserving_2013b}, or $k$-nearest neighbor classification~\cite{knn, knnwithgc}.
Moreover, custom protocols alternate between different
secret-sharing schemes as required by the specific computation being implemented. 
Of course, the transformations between secret-sharing schemes must themselves
be implemented by secure protocols\footnote{This aspect was thoroughly
investigated in~\cite{demmler_aby_2015}.}.
This point is especially relevant to DNN training, as it amounts to
a sequence of linear operations (which are naturally represented as arithmetic circuits)
interleaved with evaluations of non-linear activation functions such as the $\relu$ (which are naturally represented as Boolean circuits). 

Some MPC frameworks work in the so-called \emph{pre-processing model} (see~\cite{mohassel2017secureml}), where computation is split into a data-independent offline phase and a data-dependent online phase. Random values useful for multiplication can be generated offline, and then used for fast secure multiplication online.
In this paper we consider total time, removing the assumption of an offline phase. This is not a fundamental limitation, as we have variants of our protocols that work in the pre-processing model as explained later.

\myparagraph{Notation for Secret-Sharing.} In this work, we focus on the two-party computation setting, which excludes solutions that rely on a majority of honest parties~\cite{bogdanov2008sharemind, mohassel2018aby}. Inspired by work on function-specific protocols, in this work we employ both Boolean sharing and additive sharing. 
We start by fixing two parties $\A$ and $\B$. 
We denote the \emph{Boolean-share} of $x \in \{0,1\}$ held by $\A$ as $\xorshare{x}_1$, and $\xorshare{x}_2$ for $\B$.
In Boolean-sharing, the shares satisfy $\xorshare{x}_1 = x\oplus \xorshare{x}_2$, where $\xorshare{x}_2$ is a random bit, and $\oplus$ signifies the \texttt{XOR} operation. 
We denote the \emph{additive-share} of integer $y \in \mathbb{Z}_q$ held by $\A$ as $\share{y}_1$, and $\share{y}_2$ for $\B$. Here $\share{y}_1 = y - \share{y}_2$, with random $\share{y}_2 \in \mathbb{Z}_q$. In practice, $q$ is $2^{\sigma}$, for $\sigma\in\{8,16,32,64,128\}$, as these are word lengths offered in common architectures. 

\myparagraph{Garbled Circuits.}
In many of our protocols, we use Yao's Garbled Circuits~\cite{yao1986generate} as a subprotocol. In this protocol, the computation is represented as a Boolean circuit. We will not introduce the protocol in detail here, and instead,
refer the reader to~\cite{DBLP:journals/joc/LindellP09} for a detailed presentation and security analysis. The relevant observation
to our work is the fact that the running time of a Garbled Circuits protocol is a function of the number of non-XOR gates in a
circuit (this is thanks to the Free-XOR technique~\cite{DBLP:conf/icalp/KolesnikovS08}). Consequently, designing efficient, largely XOR circuits is crucial in this setting.
In fact, it is so important that previous works have used digital synthesizers for this task~\cite{DBLP:conf/ccs/DemmlerDKS0Z15}. 

\myparagraph{Oblivious Transfer.} 
OT is a cryptographic primitive involving two parties: a {\em Chooser}, and a {\em Sender}. The Sender holds two messages $m_0, m_1$, and the Chooser holds a Boolean value $b$. After the execution of OT, the Chooser learns $m_b$, i.e. the Sender's message corresponding to their Boolean value. From the privacy perspective, an OT protocol is correct if it guarantees that (i) the Chooser learns nothing about $m_{1-b}$, and (ii) the Sender learns nothing about $b$. As common in MPC, this is formalized in the simulation framework (see~\cite{goldreichbook} for details).

OT is a basic primitive in MPC. In fact, any function can be evaluated securely using only an OT protocol~\cite{DBLP:conf/stoc/GoldreichMW87} and, moreover OT is a crucial component of Yao's Garbled Circuits.
A remarkable advancement in the practicality of OT protocols was the discovery of OT extension~\cite{DBLP:conf/crypto/IshaiKNP03}. This protocol allows one to compute a small number of OTs, say $128$, and then bootstrap them to execute many fast OTs. Since then, optimizations in both the base OTs and the OT extension procedure~\cite{asharov2013more} have led to implementations that can perform over ten million OT executions in under one second~\cite{emp-toolkit}.

In our protocols, we employ a more efficient primitive---Correlated Oblivious Transfer (COT). COT was introduced in~\cite{asharov2013more} alongside with an efficient COT Extension protocol that uses roughly half the communication than the general OT Extension. COT is a particular case of OT where the Sender does not get to choose its messages, but instead chooses a function $f$ relating the two messages $m_0$ and $m_1$, as $m_0  = f(m_1)$. This functionality is enough for important applications of OT in MPC such as Garbled Circuits and OT-based triplet generation (see~\cite{demmler_aby_2015}). Below we define the flavour of COT that we need in our application, following the notation from~\cite{asharov2013more}.

\begin{definition}[$m\times\cot{\ell}$]
Let $\vec{f} = (f_i)_{i \in [m]}$, be a sequence of {\em correlation functions}, each with signature $f_i:\{0,1\}^\ell \mapsto \{0,1\}^\ell$ held by party $\A$ (the Sender), and let $\vec{w}\in\{0,1\}^\ell$ be a sequence of {\em choice bits} held by party $\B$ (the Chooser).
After an execution of $m\times\cot{\ell}(\vec{f}, \vec{w})$, $\A$ obtains $m$ \emph{random} vectors $(\vec{x}_i\in\{0,1\}^\ell)_{i=1}^m$, and $\B$ obtains $m$ vectors $\b(\vec{y}_i\in\{0,1\}^\ell)_{i=1}^m$ such that $\forall i\in[m]: y_i = (\neg \vec{w}_i)\cdot \vec{x}_i + \vec{w}_i\cdot f_i(\vec{x}_i)$.
\end{definition}

\subsection{Threat Model}
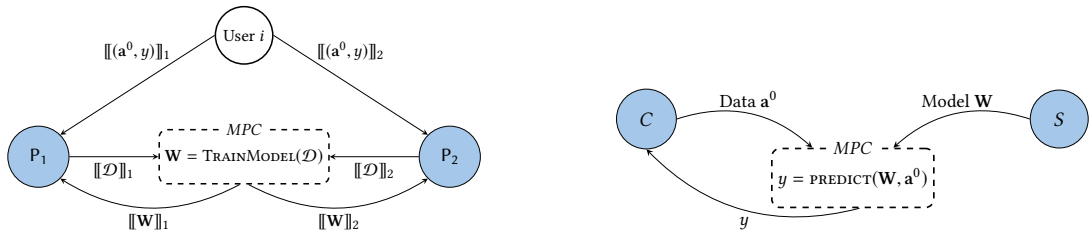
\begin{figure*}
\begin{minipage}[b]{0.45\textwidth}
\flushright
  \resizebox{0.8\textwidth}{!}{
  \begin{tikzpicture}

    \node[circle, thick, draw] (u3) at (5, 0) {User $i$};

    \node[circle, draw, fill=distinct11!66, inner sep=7pt] (p1) at ($(u3.south)+(-3.8,-1.7)$) {
        {\large $\A$}
    };

    \node[circle, draw, fill=distinct11!66, inner sep=7pt] (p2) at ($(u3.south)+(3.8,-1.7)$) {
        {\large $\B$}
    };

    \path[arrow] ($(u3.west)$) edge [] node[above=.35] {$\share{(\ab^0,y)}_1$} ($(p1.north east)$);
    \path[arrow] ($(u3.east)$) edge [] node[above=.35] {$\share{(\ab^0,y)}_2$} ($(p2.north west)$);

    \node[roundbox, draw, dashed, thick, fill=white, minimum size=1cm] (mpc) at ($(u3.south)+(0,-1.7)$) {
         $\Wb = \textsc{TrainModel}(\mathcal{D})$
    };

    \node[fill=white] (mpctext) at ($(u3.south)+(0,-1.2)$) {$MPC$};

    \path[arrow] ($(p1.east)$) edge [] node[below] {$\share{\mathcal{D}}_1$} ($(mpc.west)$);
    \path[arrow] ($(p2.west)$) edge [] node[below] {$\share{\mathcal{D}}_2$} ($(mpc.east)$);

    \path[arrow] ($(mpc.south west)+(1.5, 0)$) edge [bend left] node[below] {$\share{\Wb}_1$} ($(p1.south east)$);
    \path[arrow] ($(mpc.south east)+(-1.5, 0)$) edge [bend right] node[below] {$\share{\Wb}_2$} ($(p2.south west)$);

  \end{tikzpicture}
  }
\end{minipage}~~~~~~~~~~
\begin{minipage}[b]{0.45\textwidth}
\flushright
  \resizebox{0.8\textwidth}{!}{
  \begin{tikzpicture}

    \node[] (u3) at (5, 0) {};

    \node[circle, draw, fill=distinct11!66, inner sep=7pt] (p1) at ($(u3.south)+(-3.5,-1)$) {
        {\large $C$}
    };

    \node[circle, draw, fill=distinct11!66, inner sep=7pt] (p2) at ($(u3.south)+(3.5,-1)$) {
        {\large $S$}
    };

    \node[roundbox, draw, dashed, thick, fill=white, minimum size=1cm] (mpc) at ($(u3.south)+(0,-2)$) {
         $y = \textsc{predict}(\Wb, \ab^0)$
    };

    \node[fill=white] (mpctext) at ($(u3.south)+(0,-1.5)$) {$MPC$};

    \path[arrow] ($(p1.east)$) edge [bend left] node[above] {Data $\ab^0$} ($(mpc.north west)+(.7, 0)$);
    \path[arrow] ($(p2.west)$) edge [bend right] node[above] {Model $\Wb$} ($(mpc.north east)-(.7, 0)$);

    \path[arrow] ($(mpc.south west)+(1.5, 0)$) edge [bend left] node[below] {$y$} ($(p1.south)$);

  \end{tikzpicture}
  }
\end{minipage}
    \caption{(Left) Training in the two-server model of MPC: (i) Each user $i$ shares their labeled data $(\ab^0_i, y_i)$ across two servers $\A$ and $\B$, by giving one share $\share{(\ab^0_i, y_i)}_j$ to each $\mathsf{P}_j$.
  (ii) Each $\mathsf{P}_j$ compiles their share of a training dataset $\mathcal{D}$, by simply accumulating all shares received from users. Finally, (iii) $\A$ and $\B$ engage in a multi-party computation, in which $\mathcal{D}$ is securely reconstructed, and subsequently used to train a model $\Wb$, from which each server gets a share.
     (Right) Private prediction using MPC: A client $C$ and a server $S$ engage in an multi-party computation protocol for the client to obtain $y$ (the prediction of the server's model $\Wb$ for their client's data $\ab^0$) without the parties disclosing anything about $\Wb$ and $\ab^0$ to each other.
   }
  \label{fig:setting}
\end{figure*}
In the two-party model of MPC, the training procedure is outsourced to two servers. This framework has been used in several previous works~\cite{mohassel2017secureml, nikolaenko_privacy-preserving_2013, nikolaenko_privacy-preserving_2013b, DBLP:conf/sp/NayakWIWTS15, petsregression, blindjustice}. It works by first secret-sharing the training dataset $\mathcal{D}$ across the servers. This is depicted in Figure~\ref{fig:setting} (Left): users secret-share their values $(\ab^0_i,y_i)$ across two non-colluding servers, which run the training procedure and obtain the resulting secret-shared model: $\share{\Wb}_1, \share{\Wb}_2$. 
Note that the scenario where two organizations collaborate to build a model of their respective private data in a privacy-preserving way is a particular case of this setting: this corresponds to the stage after the users have shared their data. We present fast MPC protocols to implement $\textsc{TrainModel}(\mathcal{D})$ from Figure~\ref{fig:setting} (Left). 
Thus, all our protocols are presented as two-party protocols between parties $\A$ and $\B$ via input and output additive-shares. Our protocols are secure in the semi-honest model for secure computation, as in \cite{mohassel2017secureml}.

Alongside private training, an important related problem is private prediction. This is depicted in Figure~\ref{fig:setting} (Right). Specifically, a client $C$ has private data $\ab^0$ for which they wish to have a private prediction $y$, via the private weights $\Wb$ of server $S$. In our setting, 
a fast protocol for two-party training in MPC immediately yields a fast protocol for private prediction, as prediction corresponds to the forward pass in training (i.e., Algorithm~\ref{alg:fp_fp}).

\section{Deep Learning for MPC}
\label{section:ml}

\label{sec:ml}

In this section, we describe new methods to optimize DNNs that were developed alongside our protocols (in Section~\ref{section:crypto}). Our first insight is that recent work on training deep networks in fixed-point \cite{wu2018training} can be leveraged for crypto-friendly training. Namely, while this work was originally intended for embedded devices, it contains useful primitives such as repeated quantization to low fixed-point precisions to stabilize optimization. However, out-of-the-box, this work is unsuited for privacy-preserving protocols. We make the following modifications: 
(a) We ternarize the network weights: $\Wb \in \{-1,0,1\}$ during the forward and backward passes. This will allow matrix multiplication with $\Wb$ to be phrased as repeated 1-out-of-2 oblivious transfers (described in Section~\ref{section:crypto}); 
(b) We construct an MPC-friendly quantization function for the weight gradients in the backward pass, replacing a biased coin flip and truncation steps with a saturation-free quantization, without loss in accuracy;
(c) We replace the backward pass normalization operation, division by the closest-power-of-two, with the \emph{next}-power-of-two. While a seemingly small change, the latter operation has a very efficient circuit implementation~\cite{DBLP:books/aw/Warren2013}  that we leverage in our secure protocols in the next section. Further, we observe empirically in Section~\ref{sec:experiments} that this change also has no effect on accuracy. Ultimately, these changes will only speed up the computation of a training iteration in a secure protocol. If training is based on stochastic gradient descent, many iterations will still be necessary for convergence. To address this, we design a new fixed-point adaptive gradient algorithm. It is inspired by a state-of-the-art floating-point adaptive gradient method, AMSgrad \cite{reddi2018convergence}. This optimization allows us to achieve the best accuracy to date on all $5$ datasets we consider for DNNs trained in fixed-point (in Figure~\ref{fig:MPCTricks}). In this section, we describe the work of \cite{wu2018training}, our changes (a-c) mentioned above, and our new fixed-point adaptive gradient training procedure.

\subsection{Deep Models in Fixed-Point}
\label{section:discrete-learning}
The work by \citet{wu2018training} describes an efficient optimization procedure for DNNs operating entirely on fixed-point encodings of real numbers. We describe it in detail here.

\myparagraph{Quantize gradients, and quantize frequently.}
The first idea of WAGE~\cite{wu2018training} is to introduce functions $Q_{w}, Q_{a}, Q_{g}, Q_{e}$ that quantize the \emph{weights} $\Wb$, \emph{activations} $\ab$, \emph{weight gradients} $\Gb$, \emph{activation gradients} $\eb$ to a small, finite set of fixed-point numbers. While previous work \cite{kim2016bitwise,hubara2016binarized,rastegari2016xnor,alistarh2017qsgd,li2017training,wen2017terngrad,bernstein2018signsgd,de2018high} had already introduced the idea of functions $Q_{w},Q_{a}$ to quantize weights and/or activations in the forward pass, they required the weights $\Wb$ or gradients $\Gb,\eb$ to be represented in floating-point in the backward pass, in order to optimize accurately. 
In \citet{wu2018training}, all the quantization functions take fixed-point numbers with some precision $p$, i.e., $v_{(p)} = v/2^p$ and find the nearest number with another precision $q$, i.e., $v_{(q)}$:
\begin{align}
v_{(q)} = N(v_{(p)}, q) = \frac{\lfloor \frac{v}{2^p}2^q \rceil}{2^q} 
\end{align}
where $\frac{a}{2^p}2^q$ is in practice either a division of a multiplication by $2^{|p-q|}$ depending on whether $p>q$. 
Additionally, \citet{wu2018training} introduce a saturation function $S(\cdot)$, to yield $Q(\cdot)$ as:
\begin{align}\label{eq:Q}
v_{(q)} = Q(v_{(p)}, q) = S(N(v_{(p)}, q),q). 
\end{align}
where $S(x, q) = \min(\max(x, -1+2^{-q}), 1-2^{-q})$ saturates any $x$ to be within $[-1+2^{-q}, 1-2^{-q}]$.

\myparagraph{Weight Quantization $Q_{w}$.}
The first function $\Wb_{(p_w)} = Q_{w}(\Wb_{(p)}, p_w)$ takes as input a $p$-precision fixed-point weight $\Wb_{(p)}$ and returns the closest $p_w$-precision weights $\Wb_{(p_w)}$ using the above function:
\begin{align}
\Wb_{(p_w)} = Q_w(\Wb_{(p)}, p_w) = Q(\Wb_{(p)}, p_w) 
\end{align}

\myparagraph{Activation quantization $Q_{a}$.}
The second function $\ab_{(p_a)} = Q_{a}(\ab_{(p)}, p_a)$ is almost identical except it introduces an additional scaling factor $2^{-\alpha}$ as follows:
\begin{align}
\ab_{(p_a)} = Q_{a}(\ab_{(p)}, p_a) =  Q(\frac{\ab_{(p)}}{2^{\alpha}}, p_a).
\end{align}
Where $\alpha$ is a fixed integer determined by the dimensionality of $\Wb$, see \cite{wu2018training} eq.~(7). The intuition is that this is a simple form of normalization (used in the forward pass), and the authors describe a technique to set $\alpha$ based on the layer size.

\myparagraph{Activation gradient quantization $Q_{e}$.}
For the backward pass the first gradient we must consider is that of the activations, which we call $\eb$. The strategy to quantize $\eb_{(p_e)} = Q_{e}(\eb_{(p)}, p_e)$, is similar to quantizing the activations. Except here the scaling factor depends on the largest value of $\eb$:
\begin{align}
\eb_{(p_e)} = Q_{e}(\eb_{(p)}, p_e) = Q(\frac{\eb_{(p)}}{2^{\texttt{cpow}(\max\{| \eb_{(p)} |\})}}, p_e) \label{eq:e_quant}
\end{align}
where \texttt{cpow}$(x) = 2^{\lfloor \log_2(x) \rceil}$ is the closest power of $2$ to $x$. The intuition here is again that normalization helps, but it cannot be constant as the gradient magnitude can potentially fluctuate. Thus normalization needs to be data-specific to stabilize training.

\begin{algorithm2e}[t]
\DontPrintSemicolon
\SetKwComment{Comment}{{\scriptsize$\triangleright$\ }}{}
\caption{Fixed-Point Forward Pass (\texttt{Forward})}\label{alg:fp_fp}
        \KwIn{Fixed-point weights $\{\Wb^l_{(p_w)}\}_{l=1}^L$, \\ Fixed-point data sample $\ab^0_{(p_a)}$, \\ Layer-wise normalizers $\{\alpha_l\}_{l=1}^L$}
        \KwOut{Fixed-point activations $\{\ab^l_{(p_a)}\}_{l=1}^L$}
\BlankLine
\begin{minipage}{\hsize}
  \begin{algorithmic}[1]
      \FOR{$l = 1,\ldots,L$}
      \STATE $\ab^l = f(\Wb^l_{(p_w)} \ab^{l-1}_{(p_a)})$ \Comment*[r]{{\scriptsize pass through layer}} 
      \STATE $\ab^l_{(p_a)} = Q_{a}(\ab^l , p_a)$ \Comment*[r]{{\scriptsize activations precision $p_a$}} 
      \ENDFOR
\end{algorithmic}
\end{minipage}
\end{algorithm2e}

\begin{algorithm2e}[t]
\DontPrintSemicolon
\SetKwComment{Comment}{{\scriptsize$\triangleright$\ }}{}
\caption{Fixed-Point Backward Pass (\texttt{Backward})}\label{alg:fp_bp}
        \KwIn{Fixed-point weights: $\{\Wb^l_{(p_w)}\}_{l=1}^L$, \\Fixed-point activations (act) $\{\ab^l_{(p_a)}\}_{l=1}^L$, \\Fixed-point label $y_{(p_a)}$, \\
        Activation function $f$, \\
        Saturation function $S$}
        \KwOut{Fixed-point gradients $\{\Gb^l\}_{l=1}^L$}
\BlankLine
\begin{minipage}{\hsize}
  \begin{algorithmic}[1]
    \STATE $\eb^L = \frac{\partial \ell(\ab^L_{(p_a)}, y_{(p_a)})}{\partial \ab^L_{(p_a)}}$ \Comment*[r]{{\scriptsize act.\ gradient}}
      \FOR{$l = L,\ldots,1$}
      \STATE $\eb^l_{(p_e)} = Q_{e}(\eb^l, p_e)$ \Comment*[r]{{\scriptsize act.\ gradient precision $p_e$}} 
      \STATE $\eb^{l-1} = \Wb_{(p_w)}^l \Big(\eb^{l}_{(p_e)} \circ \frac{\partial \ab^{l}_{(p_a)} }{\partial f} \circ \frac{ \partial \ab^{l}_{(p_a)}}{\partial S}\Big)$  \Comment*[r]{{\scriptsize act.\ gradient}}
      \STATE $\Gb^{l} = {\ab^{l-1}_{(p_a)}}^\top\eb^l_{(p_e)}$ \Comment*[r]{{\scriptsize weight gradient}}
      \ENDFOR
\end{algorithmic}
\end{minipage}
\end{algorithm2e}

\myparagraph{Weight gradient quantization $Q_{g}$.}
The second gradient is the weight gradient $\Gb$, quantized by $\Gb_{(p_g)} = Q_{g}(\Gb_{(p)}, p_g)$, as follows: 
\begin{align}
\Gb_{(p_g)} = Q_{g}(\Gb_{(p)}^{n}, p_g) = \frac{\texttt{sign}(\Gb_{(p)}^{n})}{2^{p_g - 1}} \Bigg[ \lfloor | \Gb_{(p)}^{n} | \rfloor + \textrm{B}\Big(| \Gb_{(p)}^{n} | - \lfloor | \Gb_{(p)}^{n} | \rfloor\Big)   \Bigg]. \label{eq:bad_wg}
\end{align}
where $\Gb_{(p)}^n = \eta \Gb_{(p)}/2^{\texttt{cpow}(\max\{| \Gb_{(p)} |\})}$ is a normalized version of $\Gb_{(p)}$ that is also multiplied by the learning rate $\eta$, $\texttt{sign}(a)$ returns the sign of $a$, and $B(p)$ draws a sample from a Bernoulli distribution with parameter $p$. This function normalizes the gradient, applies a randomized rounding, and scales it down. The idea behind this function is that additional randomness (combined with the randomness of selecting data points stochastically during training) should improve generalization to similar, unseen data.
\subsection{\emph{Cryto-Friendly} Deep Models in Fixed-Point}
We propose three changes to the quantization functions that do not affect accuracy, but make them more suitable for MPC.

\myparagraph{Weight quantization $Q_{w}$.}
We fix $p_w\!=\!1$ and set $Q_w(\Wb_{(p)}, p_w) = 2Q(\Wb_{(p)}, p_w)$ so that our weights are ternary: $\Wb \in \{-1,0,1\}$ as mentioned above.

\myparagraph{Activation gradient quantization $Q_{e}$.}
We make the following change to eq.~\eqref{eq:e_quant} to $Q_{e}(\eb_{(p)}, p_e) = Q(\eb_{(p)}/2^{\texttt{npow}(\max\{| \eb_{(p)} |\})}, p_e)$, where $\texttt{npow}(x) = 2^{\lceil \log_2(x) \rceil}$ is the \emph{next} power of $2$ after $x$. This is faster than the closest power of $2$, $\texttt{cpow}(x)$ as the latter also needs to compute the previous power of $2$ and compare them to $x$ to find the closest power.

\myparagraph{Weight gradient Quantization $Q_{g}$.}
We introduce a different quantization function than \cite{wu2018training} which is significantly easier to implement using secure computation techniques: 
\begin{align}
\Gb_{(p_g)} = Q_{g}(\Gb_{(p)}, p_g) = N(\frac{\Gb_{(p)}}{2^{\texttt{npow}(\max\{| \Gb_{(p)} |\})}}, p_g).
\end{align}
We find that the original quantization function in eq.~\eqref{eq:bad_wg} 
needlessly removes information and adds unnecessary overhead to a secure implementation.

\begin{algorithm2e}[t]
\DontPrintSemicolon
\SetKwComment{Comment}{{\scriptsize$\triangleright$\ }}{}
\caption{Fixed-Point SGD Optimization ($\Delta_{\textrm{sgd}}$)}\label{alg:sgd}
    \KwIn{Fixed-point weights $\{\Wb^l_{(\overline{p}_w)}\}_{l=1}^L$, \\ Fixed-point data $\mathcal{D}=\{(\ab^0_{(p_a),i},y_{(p_a),i}\}_{i=1}^n$, \\ Learning rate $\eta$ }
        \KwOut{Updated weights $\{\Wb^l_{(\overline{p}_w)}\}_{l=1}^L$}
\BlankLine
\begin{minipage}{\hsize}
  \begin{algorithmic}[1]
  \FOR{$t = 1,\ldots,T$}
    \STATE $(\ab^0_{(p_a)}, y_{(p_a)}) \sim \mathcal{D}$
    \STATE $\{ \Wb^l_{(p_w)} = Q_{w}(\Wb^l_{(\overline{p}_w)}, p_w) \}_{l=1}^L$
    \STATE $\{\ab^l_{(p_a)}\}_{l=1}^L =$ \texttt{Forward}$(\{\Wb  ^l_{(p_w)}\}_{l=1}^L, \ab^0_{(p_a)})$ 
    \STATE $\{\Gb^l\}_{l=1}^L =$ \texttt{Backward}$(\{\Wb^l_{(p_w)}, \ab^l_{(p_a)}\}_{l=1}^L, y_{(p_a)})$
    \STATE $\{\Gb^{l}_{(p_g)} =  Q_{g}(\Gb^l, p_g) \}_{l=1}^L$\Comment*[r]{{\scriptsize weight gradient precision $p_g$}} 
    \STATE $\{\Wb_{(\overline{p}_w)}^l = S(\Wb^l_{(\overline{p}_w)} - \eta \Gb_{(p_g)}^l, \overline{p}_w)\}_{l=1}^L$ 
  \ENDFOR
\end{algorithmic}
\end{minipage}
\end{algorithm2e}

\myparagraph{Forward/Backward passes.}
Algorithm~\ref{alg:fp_fp} corresponds to prediction in DNNs, and Algorithms~\ref{alg:fp_fp},~\ref{alg:fp_bp},~\ref{alg:sgd} describe the training procedure. Apart from the quantization functions note that in the backward pass (Algorithm~\ref{alg:fp_bp})
we take the gradient of the activation with respect to the activation function: $\partial \ab^{l}_{(p_a)} / \partial f$ and the saturation function: $\partial \ab^{l}_{(p_a)} / \partial S$. Algorithm~\ref{alg:sgd} describes a stochastic gradient descent (SGD) algorithm for learning fixed-point weights inspired by \cite{wu2018training}. We keep around two copies of weights in different precisions $p_w$ and $\overline{p}_w$. The weights $\Wb_{(p_w)}$ are ternary and will enable fast secure forward (Algorithm~\ref{alg:fp_fp})
and backward passes. The other weights $\Wb_{(\overline{p}_w)}$ are at a higher precision and are updated with gradient information. We get the ternary weights $\Wb_{(p_w)}$ by quantizing the weights $\Wb_{(\overline{p}_w)}$ in line 3. 
Note that the forward and backward passes of convolutional networks can be written using the exact same steps as Algorithms~\ref{alg:fp_fp} and \ref{alg:fp_bp} where weights are reshaped to take into account weight-sharing. 
Similarly, for residual networks the only differences are: (a) line 2 in the forward pass changes; and 
(b) line 4 in the backward pass has an added term from prior layers.

\subsection{A Fixed-Point Adaptive Gradient Algorithm}
One of the state-of-the-art adaptive gradient algorithms is AMSgrad \cite{reddi2018convergence}. However, AMSgrad includes a number of operations that are possibly unstable in fixed-point: (i) the square-root of a sum of squares, (ii) division, (iii) moving average. 
We redesign AMSgrad in Algorithm~\ref{alg:amsgrad} to get around these difficulties in the following ways. First, we replace the square-root of a sum of squares with absolute value in line 9, a close upper-bound for values $\vb \in [-1,1]$. We verify empirically that this approximation does not degrade performance.

Second, we replace division of $\hat{\Vb}_{(p_v)}$ in line 11 with division by the next power of two. Third, we continuously quantize the weighted moving sums of lines 7 and 8 to maintain similar precisions throughout. These changes make it possible to implement AMSgrad in secure computation efficiently. In the next section we describe the protocols we design to run Algorithms~\ref{alg:fp_fp}-\ref{alg:amsgrad} privately.

\begin{algorithm2e}[t]
\DontPrintSemicolon
\SetKwComment{Comment}{{\scriptsize$\triangleright$\ }}{}
\SetKwInOut{State}{State}
\caption{Fixed-Point AMSgrad Optim.\ ($\Delta_{\textrm{ams}}$)}\label{alg:amsgrad}
		\KwIn{Fixed-point weights $\{\Wb^l_{(\overline{p}_w)}\}_{l=1}^L$, \\ Fixed-point data $\mathcal{D}=\{(\ab^0_{(p_a),i},y_{(p_a),i}\}_{i=1}^n$, learning rate $\eta$}
        \KwOut{Updated weights $\{\Wb^l_{(\overline{p}_w)}\}_{l=1}^L$}
\BlankLine
\begin{minipage}{\hsize}
  \begin{algorithmic}[1]
  	\STATE Initialize: $\{\Mb^{l}, \Vb^{l}\}_{l=1}^L =0$
  	\FOR{$t = 1,\ldots,T$}
  		\STATE $(\ab^0_{(p_a)}, y_{(p_a)}) \sim \mathcal{D}$
  		\STATE $\{ \Wb^l_{(p_w)} = 2Q(\Wb^l_{(\overline{p}_w)}, p_w) \}_{l=1}^L$
  		\STATE $\{\ab^l_{(p_a)}\}_{l=1}^L =$ \texttt{Forward}$(\{\Wb	^l_{(p_w)}\}_{l=1}^L, \ab^0_{(p_a)})$ 
  		\STATE $\{\Gb^l\}_{l=1}^L =$ \texttt{Backward}$(\{\Wb^l_{(p_w)}, \ab^l_{(p_a)}\}_{l=1}^L, y_{(p_a)})$
		\STATE $\{\Gb^l = \frac{ \Gb^l }{ 2^{\texttt{npow}(\max\{| \Gb^l |\})} } \}_{l=1}^L$ \Comment*[r]{{\scriptsize scale gradients}}
		\STATE $\{\Mb^l = N(0.9 \Mb^l, p_m) + 0.1 \Gb^l \}_{l=1}^L$ \Comment*[r]{{\scriptsize weighted mean}}
		\STATE $\{\Vb^l = N(0.99 \Vb^l, p_v) + 0.01 \big|\Gb^l\big| \}_{l=1}^L$ 
		\STATE $\{ \hat{\Vb}^l = \max(\hat{\Vb}^l, \Vb^l) \}_{l=1}^L$
		\STATE $\{\Gb^l = \frac{\Mb^l}{2^{\texttt{npow}( |\hat{\Vb}^l| + \epsilon )} } \}_{l=1}^L$ \Comment*[r]{{\scriptsize history-based scaling }}
  		\STATE $\{\Gb^{l}_{(p_g)} = N(\Gb^l, p_g)\}_{l=1}^L$ \Comment*[r]{{\scriptsize weight gradient precision $p_g$}}
  		\STATE $\{\Wb_{(\overline{p}_w)}^l = S(\Wb^l_{(\overline{p}_w)} - \eta \Gb_{(p_g)}^l, \overline{p}_w)\}_{l=1}^L$ 
  	\ENDFOR
\end{algorithmic}
\end{minipage}
\end{algorithm2e}

\section{Oblivious Transfer for Secure Learning}
\label{section:crypto}
\label{sec:protocols}
In this section, we present custom MPC protocols to implement private versions of the algorithms
introduced in the previous section. Our protocols rely heavily on OT, and crucially exploit characteristics of our networks such as ternary weights and fixed-point (8-bit precision) gradients. We present our MPC protocol for neural network training by describing its components separately. First, in Section~\ref{sec:mv-mult-prot} we describe our protocol for ternary matrix-vector multiplication, a crucial primitive for training used both in the forward and backward pass. Next, in Sections~\ref{sec:fw-pass-prot} and~\ref{sec:bw-pass-prot} we describe the protocols for the forward and backward passes, respectively. These protocols use our matrix-vector multiplication protocol (and a slight variant of it), in combination with efficient garbled circuit implementations for normalization and computation of $\relu$.

\subsection{Secure Ternary Matrix-Vector Multiplication}\label{sec:mv-mult-prot}

A recurrent primitive, both in the forward and backward passes of the fixed-point neural networks from Section~\ref{section:ml} is the product $\Wb\ab$, for ternary matrix $\Wb \in \{-1, 0, 1\}^{n\times m}$ and fixed-point integer vector $\ab\in \mathbb{Z}^m_q$. Several previous works on MPC-based ML have looked specifically at matrix-vector multiplication~\cite{mohassel2017secureml, knn}. As described in Section~\ref{subsec:mpc}, multiplication is a costly operation in MPC. Our insight is that if $\Wb$ is ternary, we can replace multiplications by selections, enabling much faster protocols. More concretely, we can compute the product $\vec{z} = \mat{W}\ab$ as shown in Algorithms~\ref{alg:ternary-mv}.
\begin{algorithm2e}
\DontPrintSemicolon
\caption{Ternary-Integer Matrix-Vector Product}\label{alg:ternary-mv}
\KwIn{Matrix $\mat{W}\in\{-1, 0, 1\}^{n \times m}$ and vector $\vec{a}\in\mathbb{Z}_q^m$}
\begin{minipage}{\hsize}
  \begin{algorithmic}[float, plain]
    \STATE $\vec{z} = (0)_{i\in[n]}$        
    \FOR{$i\in [n], j\in[m]$}
    \STATE \algorithmicif\ $\mat{W_{i,j}} > 0$ \algorithmicthen\ $z_i~\mathrel{+}=~\ab_j$
    \STATE \algorithmicif\ $\mat{W_{i,j}} < 0$ \algorithmicthen\ $z_i~\mathrel{-}=~\ab_j$
    \ENDFOR
    \RETURN $\vec{z}$
\end{algorithmic}
\end{minipage}
\end{algorithm2e}

A natural choice for implementing the functionality in Algorithm~\ref{alg:ternary-mv} securely are MPC protocols that represent the computation as a Boolean circuit. In our two-party setting natural choices are garbled circuits and the GMW protocol \cite{DBLP:conf/stoc/GoldreichMW87}. In Boolean circuits, the If-Then-Else construction corresponds to a multiplexer, and a comparison with $0$ is essentially free: it is given by the sign bit in a two's complement binary encoding. However, a circuit implementation of the computation above will require $|\mat{W}|$ additions and the same number of subtractions, which need to implemented with full-adders. Whereas, additions, if computed on additive shares, do not require interaction and thus are extremely fast.

Our proposed protocol achieves the best of both worlds by combining Boolean sharing and additive sharing. To this end, we represent the ternary matrix $\mat{W}$ by two Boolean matrices $\mat{W^+}\in\{0,1\}^{n \times m}$ and $\mat{W^-}\in\{0,1\}^{n \times m}$, defined as $\Wb^+_{i,j} = 1 \Leftrightarrow \mat{W}_{i,j} = 1$ and $\Wb^-_{i,j} = 1 \Leftrightarrow \mat{W}_{i,j} = -1$. Now, the product $\mat{W} \ab$ can be rewritten as $\mat{W^+}\vec{a} - \mat{W^-}\vec{a}$. This reduces our problem from $\mat{W}\vec{a}$ with ternary $\mat{W}$ to two computations of $\mat{W}\vec{a}$ with binary $\mat{W}$. Note that we can use the same decomposition to split any matrix with inputs in a domain of size $k$ into $k$ Boolean matrices, and thus our protocol is not restricted to the ternary case.

Accordingly, at the core of our protocol is a two-party subprotocol for computing additive shares of $\mat{W}\vec{a}$, when $\mat{W}$ is Boolean-shared among the parties and $\vec{a}$ is additively shared. In turn, this protocol relies on a two-party subprotocol for computing additive shares of the inner product of a Boolean-shared binary vector and an additively shared integer vector. As a first approach to this problem, we show in Protocol~\ref{prot:bool-bcomb} a solution based on oblivious transfer. We state and prove the correctness of Protocol~\ref{prot:bool-bcomb} in Appendix~\ref{sec:appendix-proof}.

One can think of Protocol~\ref{prot:bool-bcomb} as a component-wise multiplication protocol, as we will use it also for that purpose. The only modification required is in step $6$, i.e. the local aggregation of additive shares of the result of component-wise multiplication. We could directly use this protocol to implement our desired matrix-vector multiplication functionality, and this leads to very significant improvements due to the concrete efficiency of OT Extension implementations. However, we can further optimize this to obtain our final protocol.

\begin{protocol}[t]
\DontPrintSemicolon
\caption{Boolean-Integer Inner Product}\label{prot:bool-bcomb}
        {\bf Parties:}~ $\A$ and $\B$\\
        \KwIn{Arithmetic shares of integer vector $\vec{a}\in\mathbb{Z}_q^m$ and Boolean shares of binary vector $\vec{w}\in\{0, 1\}^m$}
        \KwOut{Arithmetic shares of $z = \vec{w}^\top \vec{a}$}
\BlankLine
\begin{minipage}{\hsize}
  \begin{algorithmic}[1]
    \STATE Each $\partyi$ generates random values $(z_{i, j})_{j\in[m]}$.
    \FOR{$j\in[m]$}
    \STATE $\partyi$ sets\\\vspace{-1ex}\hspace{0.2\columnwidth}$m_{i, 0} := \xorshare{\vec{w}_j}_i\cdot \share{\vec{a}_j}_i - z_{i, j}$
    \\\hspace{0.2\columnwidth}$m_{i, 1} := \neg\xorshare{\vec{w}_j}_i\cdot \share{\vec{a}_j}_i - z_{i, j}$\\[1ex]
    \STATE $\A$ and $\B$ run $\ot(m_{1, 0}, m_{1, 1}, \xorshare{\vec{w}_j}_2)$, with $\A$ as Sender and $\B$ as Chooser, for $\B$ to obtain $z_{1, j}'$
    \STATE $\A$ and $\B$ run $\ot(m_{2, 0}, m_{2, 1}, \xorshare{\vec{w}_j}_1)$, with $\B$ as Sender and $\A$ as Chooser, for $\A$ to obtain $z_{2, j}'$
    \ENDFOR
    \STATE Each $\partyi$ sets $\share{z}_i = \sum_{j\in[m]}(z_{i,j} + z_{i,j}')$.
\end{algorithmic}
\end{minipage}
\end{protocol}

The use of OT in Protocol~\ref{prot:bool-bcomb} has similarities with the GMW protocol, and is inspired by the OT-based method due to Gilboa for computing multiplication triplets, and discussed in~\cite{demmler_aby_2015}. A similar idea was used in the protocol for computing the sigmoid function by~\citet{mohassel2017secureml}. Moreover, concurrently to this work, \citet{riazi2019xonn} have proposed OT-based protocols for secure prediction using DNNs. They propose a protocol called Oblivious Conditional Addition (OCA) that is analogous to Protocol~\ref{prot:bool-bcomb}. While their work only addresses secure prediction, our improved protocol presented in the next section is relevant in their setting as well.

\subsubsection{Optimizations: Correlated OT and Packing.}

In the previous section, our protocol assumed a standard OT functionality, but actually, we can exploit even more efficient primitives. Our optimization of Protocol~\ref{prot:bool-bcomb}, presented in Protocol~\ref{prot:bool-bcomb-cot} exploits COT in a way to implement our required inner product functionality. 
The idea behind Protocol~\ref{prot:bool-bcomb-cot} is simple: note that in Protocol~\ref{prot:bool-bcomb} parties choose their random shares $z_{i,j}$ of intermediate values in the computation, and they use them to mask OT messages. However, as we only require the $z_{i,j}$'s to be random, one could in principle let the OT choose them. Note also that the parties can choose their share of the result as a function of their inputs, which can be implemented in COT. This is done in lines $4$ and $5$ of Protocol~\ref{prot:bool-bcomb-cot}. The next Lemma states the correctness of Protocol~\ref{prot:bool-bcomb-cot}. As the protocol consists of just two executions of $m\times\cot{\ell}$ and local additions its security is trivial, while we present the proof of correctness in Appendix~\ref{app:proofbool-bcomb-cot}. Finally, note that, as in the case of Protocol~\ref{prot:bool-bcomb}, it is easy to turn Protocol~\ref{prot:bool-bcomb-cot} into a protocol for component-wise multiplication.

\begin{protocol}[t]\label{}
\DontPrintSemicolon
\caption{Boolean-Integer Inner Product via $m \!\times\! \cot{\ell}$}\label{prot:bool-bcomb-cot}
        {\bf Parties:}~ $\A$ and $\B$\\
        \KwIn{$\ell$-bit arithmetic shares of integer vector $\vec{a}\in\mathbb{Z}_q^m$ and Boolean shares of binary vector $\vec{w}\in\{0, 1\}^m$}
        \KwOut{Arithmetic shares of $z = \vec{w}^\top \vec{a}$}
\BlankLine
\begin{minipage}{\hsize}
  \begin{algorithmic}[1]
        \STATE  Each party $\partyi$ constructs a vector of correlation functions
        $\vec{f^{~i}} = (f_{i,j}(x))_{j\in[m]}$,
        \vspace{-1ex} $$f_{i,j}(x) = x - \share{\vec{w}_j}_i\cdot \share{\vec{a}_j}_i + \neg \share{\vec{w}_j}_i\cdot \share{\vec{a}_j}_i $$\vspace{-3ex}
        \STATE The parties run $m \times \cot{\ell}(\vec{f^{~1}}, \share{\vec{w}_j}_2)$ with $\A$ acting as the Sender, and $\A$ obtains $\vec{x}$ while $\B$ obtains $\vec{y}$.
        \STATE The parties run $m \times \cot{\ell}(\vec{f^{~2}}, \share{\vec{w}_j}_1)$ with $\B$ acting as the Sender, and $\B$ obtains $\vec{x'}$ while $\A$ obtains $\vec{y'}$.
        \STATE $\A$ sets $\share{z}_1 = \sum_{j\in[m]}(\share{\vec{w}_j}_1\cdot\share{\vec{a}_j}_1 - \vec{x}_j + \vec{y'}_j)$
        \STATE $\B$ sets $\share{z}_2 = \sum_{j\in[m]}(\share{\vec{w}_j}_2\cdot\share{\vec{a}_j}_2 - \vec{x'}_j + \vec{y}_j)$
\end{algorithmic}
\end{minipage}
\end{protocol}

\begin{lemma}
Let $\vec{w}$ and $\vec{a}$ be a Boolean and integer vector, respectively, shared among parties $\A, \B$.
Given an $m\times\cot{\ell}$ protocol, the two-party Protocol~\ref{prot:bool-bcomb-cot} is secure against semi-honest adversaries, and computes an additive share of the inner product $\vec{w}^\top \vec{a}$ among $\A, \B$. 
\end{lemma}

Protocol~\ref{prot:ternary-mv} achieves our goal of computing an arithmetic share of $\mat{W}\vec{a} = \mat{W^+}\vec{a} - \mat{W^-}\vec{a}$, for an $n\times m$ matrix. This is easily achieved using $2n$ calls to Protocol~\ref{prot:bool-bcomb-cot}. This translates into
$4n$ executions of $m\times\cot{\ell}$, plus {\em local} extremely efficient additions. Note that this protocol is fully paralellizable, as all the COT executions can be run in parallel.

\begin{protocol}[t]
\DontPrintSemicolon
\caption{Ternary-Integer Matrix-Vector Product}\label{prot:ternary-mv}
        {\bf Parties:}~ $\A$ and $\B$\\
        \KwIn{Arithmetic shares of integer vector $\vec{a}\in\mathbb{Z}_q^m$ and Boolean shares of binary matrices $\mat{W^+}, \mat{W^-}\in\{0, 1\}^{n,m}$}
        \KwOut{Arithmetic shares of $\vec{z} = \mat{W^+}\vec{a} - \mat{W^-}\vec{a}$}
\BlankLine
\begin{minipage}{\hsize}
  \begin{algorithmic}[1]
    \STATE $\A$ and $\B$ compute $\share{\mat{W^+}\vec{a}}$ using $n$ executions of Protocol~\ref{prot:bool-bcomb-cot}.
    \STATE $\A$ and $\B$ compute $\share{\mat{W^-}\vec{a}}$ using $n$ executions Protocol~\ref{prot:bool-bcomb-cot}.
    \STATE $\partyi$ sets  $\share{\vec{z}}_i := \share{\mat{W^+}\vec{a}}_i - \share{\mat{W^+}\vec{a}}_i$.
\end{algorithmic}
\end{minipage}
\end{protocol}

Overall, our approach exploits the fact that $\mat{W}$ is ternary without having to perform any Boolean additions in secure computations. Our experiments in Section~\ref{sec:experiments} show concrete gains over the prior state-of-the-art.

\myparagraph{Communication costs.}
Using the $m\times\cot{\ell}$ Extension protocol from \cite{asharov2013more},
the parties running Protocol~\ref{prot:bool-bcomb-cot} send $m(\tau + \ell)$ bits
to each other to compute inner products of length $m$, where $\tau$ is the security parameter ($128$ in our implementation). This results in $nm(\tau + \ell)$ bits being sent/received by each party for the whole matrix-vector multiplication protocol. In contrast, the OT-based approaches to matrix-vector multiplication entirely based on arithmetic sharing from~\cite{mohassel2017secureml} would require at least $nm\ell(\tau + \ell)$ (assuming optimizations like packing and vectorization).

\myparagraph{Packing for matrix-vector multiplication.}
While the forward pass of quotient operates over $8$-bit vectors (and thus $q\!=\!8$ in Protocol~\ref{prot:ternary-mv}), the value of $\ell$ in implementations of $m \times \cot{\ell}$ is $128$, i.e., the AES block size. However, we can exploit this to pack $128/8\!=\!16$ vector multiplications against the same matrix for the same communication and computation. This is very useful in batched gradient descent, as this results in $16$x additional savings in communication and computation. This packing optimization was also used in~\cite{mohassel2017secureml} for implementing an OT-based offline phase to matrix-vector multiplications occurring in batched gradient descent.

\subsection{Secure Forward Pass}
\label{sec:fw-pass-prot}
Our protocol for the forward pass (Algorithm~\ref{alg:fp_fp}) is a sequential composition of Protocol~\ref{prot:ternary-mv}, and a garbled circuit protocol with three components: (i) Evaluation of $\relu$, (ii) normalization by a public data-independent value $\alpha_l$ (line $3$ in Algorithm~\ref{alg:fp_fp}), and (iii) quantization function $Q(\cdot)$ in eq.~\eqref{eq:Q} in Section~\ref{section:discrete-learning}. The protocol is depicted in Figure~\ref{fig:forward-pass}. Note that we only show a forward pass for a single layer, but the protocol trivially composes sequentially with itself, as input $\vec{a}^{\ell-1}$ and output $\vec{a}^{\ell}$ are both secret-shared additively. Security follows directly from the security of the subprotocols, as their outputs and inputs are always secret-shares. 
For scalability, we describe an efficient circuit implementation of (i)-(iii). In our proposed circuit $\A$ inputs its share $\share{\vec{z}}_1$, and a random value chosen in advance that will become its share of the output, denoted $\share{\vec{a}^\ell}_1$. In the garbled circuit, $\vec{z}$ is first reconstructed (this requires $|\vec{z}|$ parallel additions). For component (i), $\relu$, we exploit the fact that the entries in $\vec{z}$ are encoded in binary in the circuit using two's complement as ${z}_i = (b_{i,k} \cdots b_{i,1})$, where $b_{i,k} = 1 \Leftrightarrow {z}_i < 0$. Hence $\relu ({z}_i) = (\neg b_{i,k}) \cdot b_{i,k}$. Note that this is very efficient, as it only requires to evaluate $k|\vec{z}|$ NOT and AND gates. The next two steps, (ii) normalization and (iii) quantization are extremely cheap, as they can be implemented with logical and arithmetic shifts without requiring any secure gate evaluation. Finally, to construct $\B$'s output we need to perform a subtraction inside the circuit to compute $\share{\vec{a}^\ell}_2 =  \vec{a}^\ell - \share{\vec{a}^\ell}_1$. Altogether this means that our forward pass requires execution of Protocol~\ref{prot:ternary-mv} and a garbled circuit protocol to evaluate a vector addition, a vector subtraction, and a linear number of additional gates. Moreover, note that this garbled circuit evaluation can be parallelized across components of the vector $\vec{z}$.

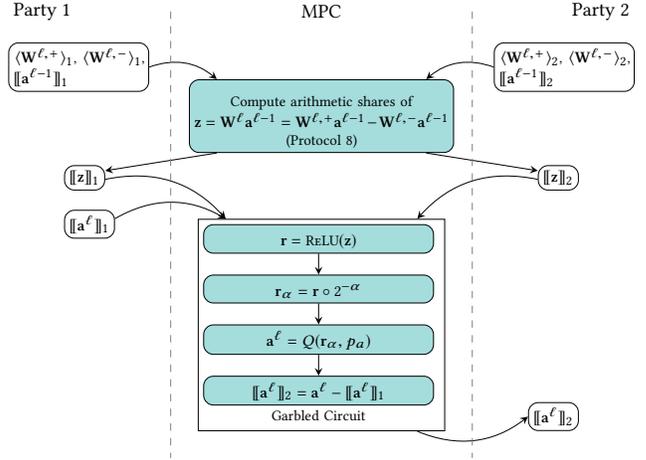
\begin{figure}\centering
  \resizebox{\columnwidth}{!}{
  \begin{tikzpicture}[scale=0.8]\scriptsize
    \node (p1) at (0, 0) {\small Party 1};
    \node (mpc) at (5, 0) {\small MPC};
    \node (p2) at (10, 0) {\small Party 2};

    \node[anchor={north west}, box p1, text width=0.22\columnwidth] (local1a)
    at ($(p1.south west)-(0,.3)$) {
     $\xorshare{\mat{W^{\ell,+}}}_{1}$, $\xorshare{\mat{W^{\ell,-}}}_{1}$, $\share{\vec{a}^{\ell -1}}_{1}$
    };
    \node[anchor={north east}, box p2, text width=0.22\columnwidth] (local1b)
    at ($(p2.south east)-(0,.3)$) {
     $\xorshare{\mat{W^{\ell,+}}}_{2}$, $\xorshare{\mat{W^{\ell,-}}}_{2}$, $\share{\vec{a}^{\ell-1}}_{2}$
    };

    \node[box mpc, text width = 0.43\columnwidth, anchor = north] (mpc1) at ($(mpc.south)+(0,-1)$) {
      \begin{center}
        Compute arithmetic shares of \\$\vec{z} = \mat{W^\ell}\vec{a}^{\ell-1} = \mat{W^{\ell, +}}\vec{a}^{\ell-1} - \mat{W^{\ell, -}}\vec{a}^{\ell-1}$\\(Protocol~\ref{prot:ternary-mv})
      \end{center}
    };

    \draw[arrow] ($(local1a.east)$) to [bend left] ($(mpc1.north west)+(.5, 0)$);
    \draw[arrow] ($(local1b.west)$) to [bend right] ($(mpc1.north east)+(-.5, 0)$);

    \node[anchor={north west}, box p1] (local2a)
    at ($(p1.south west)-(-1,2.5)$) {
     $\share{\vec{z}}_1$
    };
    \node[anchor={north west}, box p1] (local2aa)
    at ($(p1.south west)-(-1,3.3)$) {
     $\share{\vec{a}^\ell}_1$
    };
    \node[anchor={north east}, box p2] (local2b)
    at ($(p2.south east)-(1,2.5)$) {
     $\share{\vec{z}}_2$
    };

    \draw[arrow] ($(mpc1.south west)+(.5, 0)$) -- ($(local2a.north east)+(0, -.1)$);
    \draw[arrow] ($(mpc1.south east)+(-.5, 0)$) -- ($(local2b.north west)+(0, -.1)$) ;

    \node[draw, fill=white, solid, text width = 0.40\columnwidth, anchor = north] (mpc2) at ($(mpc.south)+(0,-3.5)$) {
    \begin{tikzpicture}
       \node[box mpc, solid, text width = 0.93\columnwidth] (mpc21){\vspace{-1ex}\begin{center}$\vec{r} = \relu(\vec{z})$\end{center}};
       \node[box mpc, solid, text width = 0.93\columnwidth] (mpc22) at ($(mpc21.south)+(0,-.3)$) {\vspace{-1ex}\begin{center}$\vec{r_\alpha} = \vec{r}\circ 2^{-\alpha}$\end{center}};
       \node[box mpc, solid, text width = 0.93\columnwidth] (mpc23) at ($(mpc22.south)+(0,-.3)$) {\vspace{-1ex}\begin{center}$\vec{a^\ell} = Q(\vec{r_\alpha},p_a)$\end{center}};
       \node[box mpc, solid, text width = 0.93\columnwidth] (mpc24) at ($(mpc23.south)+(0,-.3)$) {\vspace{-1ex}\begin{center}$\share{\vec{a^\ell}}_2 = \vec{a^\ell} - \share{\vec{a^\ell}}_1$\end{center}}; 

    \draw[arrow, solid] ($(mpc23.south)$) -- ($(mpc24.north)$) ; 
    \draw[arrow, solid] ($(mpc22.south)$) -- ($(mpc23.north)$) ;
    \draw[arrow, solid] ($(mpc21.south)$) -- ($(mpc22.north)$) ;

    \node[text width = 0.93\columnwidth] (text) at ($(mpc24.south)+(0,-.1)$) {\vspace{-3ex}\begin{center}Garbled Circuit\end{center}}; 
    \end{tikzpicture}
    
    };

    \draw[arrow] ($(local2a.south east)+(0, .2)$) to [bend left] ($(mpc2.north west)+(.5, 0)$);
    \draw[arrow] ($(local2aa.east)+(0, .1)$) to [bend left] ($(mpc2.north west)+(.5, 0)$);
    \draw[arrow] ($(local2b.south west)+(0, .2)$) to [bend right] ($(mpc2.north east)+(-.5, 0)$);

    \node[anchor={north east}, box p2] (local3b)
    at ($(p2.east)-(1,7)$) {
     $\share{\vec{a^\ell}}_2$
    };

    \draw[arrow] ($(mpc2.south east)+(-.5, 0)$) to [bend right] ($(local3b.west)$) ;

    \draw[dashed, gray] ($(p1)+(2.3, 0)$) -- ($(p1)+(2.3, -8)$);
    \draw[dashed, gray] ($(p2)-(2.3, 0)$) -- ($(p2)-(2.3, 8)$);
  \end{tikzpicture}
  }
  \caption{Our MPC protocol for private prediction (forward pass). We compose three protocols to evaluate one layer of the form $f(\mat{W}\vec{a})$, with ternary $\mat{W}$, and where $f = \relu$.} 
    \label{fig:forward-pass}
\end{figure}

\subsection{Secure Backward Pass}
\label{sec:bw-pass-prot}
Figure~\ref{fig:bw-pass} shows our protocols for the backward pass (Algorithm~\ref{alg:fp_bp}). 
Analogous to the forward pass, we depict the pass for one layer, and observe that all its inputs and outputs are secret-shared. Hence, it can be easily composed sequentially across layers, and with the forward pass, and its security follows trivially from the security of each of the subprotocols.

The protocol works via a sequence of subprotocols. 
Each of them produces a result shared among the parties $\A, \B$ either as a Boolean-share or an additive share. As in the forward pass, this protocol leverages Protocol~\ref{prot:ternary-mv}, as well as Protocol~\ref{prot:bool-bcomb-cot} for component-wise multiplication. Recall that the goal of the backward pass is to recompute ternary weight matrices $\{\mat{W}^\ell\}_{\ell =1}^L$ by means of a gradient-based procedure. 
As $\{\mat{W}^\ell\}_{\ell =1}^L$ is represented in our MPC protocol by pairs of binary matrices, the protocol to be run for each layer $\ell$ takes as inputs the Boolean-shares of such matrices, i.e., $\mat{W^{\ell, -}}$ and $\mat{W^{\ell, +}}$, from each party. 
Moreover, the parties contribute to the protocol arithmetic shares of the input to each layer $a^\ell$ computed in the forward pass, as well as the target values $\vec{y}$.

The first step of the backward pass is a data-dependent normalization of the activation gradient $\vec{e}^\ell$, followed by the quantization step that we described in the forward pass, as shown in Figure~\ref{fig:bw-pass}. 
We now describe the design of the Boolean circuit used to compute these steps.

\myparagraph{Normalization by the infinite norm.}
Our goal is to design a (small) circuit that, given $\vec{e}$, computes $\vec{\vec{e}}/2^{\texttt{npow}(\max\{|\vec{e}|\})}$. A naive circuit would compute the absolute value of every entry $|\vec{e}|$, compute the maximum value $\max\{|\vec{e}|\}$, compute $2^{\lceil \log_2(\max\{|\vec{e}|\}) \rceil}$, and finally compute a division. However, computing exponentiation and logarithm in a garbled circuit would be prohibitively expensive. Such circuits are large, and we have to do this computation 
in each layer $\ell$, as many times as the number of total iterations. 
To overcome this we apply two crucial optimizations: (i) approximate $\max\{|\vec{e}|\}$ by bitwise OR of all values $|\vec{e}_i|$, and (ii) compute $\texttt{npow}$ with an efficient folklore procedure for obtaining the number of leading zeros in a binary string. This requires only $b$ OR gates and  $\log(b)$ arithmetic right shits and additions, where $b$ is the bitwidth of the entries of $\vec{e}$ ($8$ in our applications) \cite{DBLP:books/aw/Warren2013}. Also note that the division can be computed as an arithmetic right shift. An important remark is that the denominator is a private value in the circuit, which means that, although we can use right shift for division, our circuit needs to first compute all possible right shifts and then select the result according to the value of $\texttt{npow}(\max\{|\vec{e}|\}))$. This involves a subcircuit linear in $b$, and since in our case $b\!=\!8$ we once again benefit from having small bitwidth, by trading a small overhead for costly divisions, logarithms, and exponentiations.

\myparagraph{Derivatives of $\relu$ and saturation.}
Computing derivatives of $\relu$ and saturation $S(\cdot)$ can be done efficiently in a Boolean representation, as they lie in $\{0,1\}$. Specifically, computation only involves extracting the sign bit for $\relu$, and ANDing a few bits for saturation. Ultimately we need to compute $\vec{e}\circ \vec{d}$, for $\vec{d} = \relu'(\vec{a}) \circ S'(\vec{a})$ (line 4 in Algorithm~\ref{alg:fp_bp}). Not that this is a Boolean combination of integers, so we can alternate between Boolean and arithmetic shares and compute it with Protocol~\ref{prot:bool-bcomb-cot} (for component-wise product). We can further optimize the procedure by computing $\xorshare{\vec{d}}_i$ in the forward pass, since $\relu$ is already used there. This is commonly done in ML implementations in the clear.

The remainder of the backward pass involves (i) computing $\vec{e}^{l-1}$ (the rest of line 4 in Algorithm~\ref{alg:fp_bp}), for which we use Protocol~\ref{prot:ternary-mv}, and (ii) an outer product between $8$-bit vectors ${\ab^{l-1}}^\top\eb$ (line 5). For (ii) we use a vectorized version of the OT-based multiplication protocol presented originally in~\cite{DBLP:conf/crypto/Gilboa99}, and used in~\cite{demmler_aby_2015,mohassel2017secureml}.

Overall, this makes our backward pass very efficient, involving three small garbled circuits (two can be parallelized), and relies heavily on oblivious transfer computations.

\myparagraph{SGD.}
To implement Algorithm~3 we need to additionally keep higher precision $8$-bit matrices $\{\mat{W}_{(\overline{p}_w)}^\ell\}_{\ell =1}^L$ as arithmetic shares. We ternarize these to obtain our weights $\Wb$, which we implement in the natural way with a small garbled circuit involving $2$ comparisons. The same operations for quantization and normalization of $\vec{e}$ can be used for $\mat{\vec{G}}$.

\myparagraph{AMSGrad.} 
Almost all of the operations in AMSGrad (Algorithm~\ref{alg:amsgrad}): quantization, normalization, saturation, and absolute value have been described as part of the previous protocols. The only addition is element-wise maximum (line 10), which we do  via a comparison of Boolean shares.

\myparagraph{Convolutional \& residual layers.} Although we have only described a fully connected layer, extending this protocol to convolutional layers is straightforward. The forward and backward passes of convolutional layers can be written using the same steps as Algorithms~\ref{alg:fp_fp} and \ref{alg:fp_bp} with weight-reshaping. And max-pooling operations are simply comparisons, efficiently implemented in Boolean shares. Similarly, for residual networks, we only introduce integer additions in the forward pass (which we can perform on additive shares) and another computation of $\vec{d}$ for the backward pass.

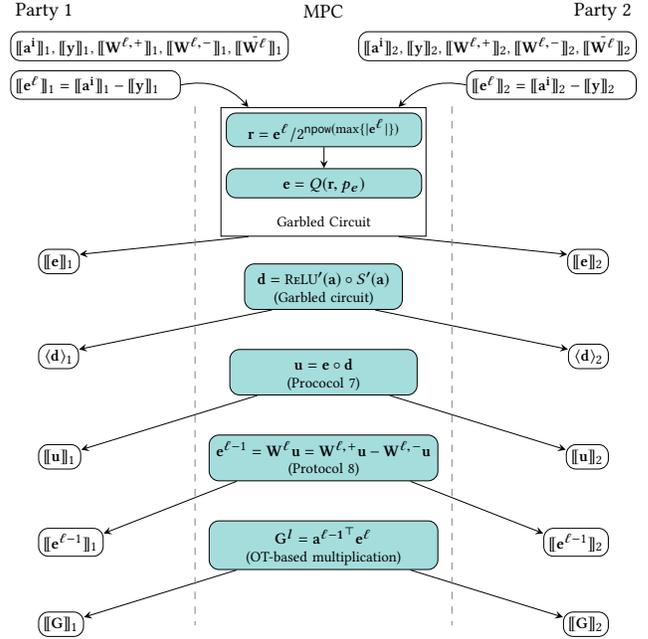
\begin{figure}\centering
  \resizebox{\columnwidth}{!}{
  \begin{tikzpicture}[scale=0.8]\scriptsize
    \node (p1) at (0, 0) {\small Party 1};
    \node (mpc) at (5, 0) {\small MPC};
    \node (p2) at (10, 0) {\small Party 2};

    \node[anchor={north west}, box p1] (local1a)
    at ($(p1.south west)-(0,.1)$) {
     $\share{\vec{a^i}}_{1}$, $\share{\vec{y}}_{1}$,
     $\share{\mat{W^{\ell, +}}}_{1}$, $\share{\mat{W^{\ell,-}}}_{1}$, $\share{\mat{\bar{W^\ell}}}_{1}$
    };
    \node[anchor={north east}, box p2] (local1b)
    at ($(p2.south east)-(0,.1)$) {
    $\share{\vec{a^i}}_{2}$, $\share{\vec{y}}_{2}$,
    $\share{\mat{W^{\ell, +}}}_{2}$, $\share{\mat{W^{\ell,-}}}_{2}$, $\share{\mat{\bar{W^\ell}}}_{2}$
    };

    \node[anchor={north west}, box p1, text width=0.27\columnwidth] (local1aa)
    at ($(p1.south west)-(0,.8)$) {
     $\share{\vec{e}^\ell}_1 =  \share{\vec{a^i}}_{1} - \share{\vec{y}}_{1}$
    };
    \node[anchor={north east}, box p2, text width=0.27\columnwidth] (local1bb)
    at ($(p2.south east)-(0,.8)$) {
     $\share{\vec{e}^\ell}_2 =  \share{\vec{a^i}}_{2} - \share{\vec{y}}_{2}$
    };

    \node[draw, fill=white, solid, text width = 0.33\columnwidth, anchor = north] (mpc1) at ($(mpc.south)+(0,-1.5)$) {
    \begin{tikzpicture}
       \node[box mpc, solid, text width = 0.95\columnwidth] (mpc21){\vspace{-1ex}\begin{center}$\vec{r} = \vec{\vec{e}^\ell}/2^{\texttt{npow}(\max\{|\vec{e}^\ell|\})}$\end{center}};
       \node[box mpc, solid, text width = 0.95\columnwidth] (mpc22) at ($(mpc21.south)+(0,-.3)$) {\vspace{-1ex}\begin{center}$\vec{e} = Q(\vec{r},p_e)$\end{center}};

    \draw[arrow, solid] ($(mpc21.south)$) -- ($(mpc22.north)$) ;

    \node[text width = 0.93\columnwidth] (text) at ($(mpc22.south)+(0,-.1)$) {\vspace{-1ex}\begin{center}Garbled Circuit\end{center}};
    \end{tikzpicture}
    
    };

    \draw[arrow] ($(local1aa.east)+(0,0)$) to [bend left] ($(mpc1.north west)+(.5, 0)$);
    \draw[arrow] ($(local1bb.west)+(0,0)$$) to [bend right] ($(mpc1.north east)+(-.5, 0)$);

    \node[anchor={north west}, box p1] (local2a)
    at ($(p1.south west)-(-.5,4)$) {
      $\share{\vec{e}}_1$
    };
    \node[anchor={north east}, box p2] (local2b)
    at ($(p2.south east)-(.5,4)$) {
     $\share{\vec{e}}_2$
    };

    \draw[arrow] ($(mpc1.south west)+(.5, 0)$) -- ($(local2a.north east)+(0, -.1)$);
    \draw[arrow] ($(mpc1.south east)+(-.5, 0)$) -- ($(local2b.north west)+(0, -.1)$) ;

    \node[box mpc, text width = 0.25\columnwidth, anchor = north] (mpc3) at ($(mpc.south)+(0,-4.3)$) {
     \vspace{-1ex}
      \begin{center} $\vec{d} = \relu'(\vec{a}) \circ S'(\vec{a})$\\(Garbled circuit)\end{center}
    };

    \node[anchor={north west}, box p1] (local3a)
    at ($(p1.south west)-(-.5,5.7)$) {
      $\xorshare{\vec{d}}_1$
    };
    \node[anchor={north east}, box p2] (local3b)
    at ($(p2.south east)-(.5,5.7)$) {
      $\xorshare{\vec{d}}_2$
    };

    \draw[arrow] ($(mpc3.south west)+(.5, 0)$) -- ($(local3a.north east)+(0, -.1)$);
    \draw[arrow] ($(mpc3.south east)+(-.5, 0)$) -- ($(local3b.north west)+(0, -.1)$) ;

    \node[box mpc, text width = 0.3\columnwidth, anchor = north] (mpc4) at ($(mpc3.south)+(0,-.7)$) {
     \vspace{-1ex}
      \begin{center}$\vec{u} = \vec{e}\circ \vec{d}$\\(Prococol~\ref{prot:bool-bcomb-cot})\end{center}
    };

    \node[anchor={north west}, box p1] (local4a)
    at ($(p1.south west)-(-.5,7.5)$) {
      $\share{\vec{u}}_1$
    };
    \node[anchor={north east}, box p2] (local4b)
    at ($(p2.south east)-(.5,7.5)$) {
     $\share{\vec{u}}_2$
    };

    \draw[arrow] ($(mpc4.south west)+(.5, 0)$) -- ($(local4a.north east)+(0, -.1)$);
    \draw[arrow] ($(mpc4.south east)+(-.5, 0)$) -- ($(local4b.north west)+(0, -.1)$) ;
    
   \node[box mpc, text width = 0.37\columnwidth, anchor = north] (mpc5) at ($(mpc4.south)+(0,-.7)$) {
     \vspace{-1ex}
      \begin{center}$\vec{e}^{\ell-1} = \mat{W^\ell}\vec{u} = \mat{W^{\ell, +}}\vec{u} - \mat{W^{\ell, -}}\vec{u}$\\(Protocol~\ref{prot:ternary-mv})\end{center}
    };

    \node[box mpc, text width = 0.37\columnwidth, anchor = north] (mpc6) at ($(mpc5.south)+(0,-.7)$) {
     \vspace{-1ex}
      \begin{center}$\mat{G}^{l} = \vec{a^{\ell-1}}^\top\vec{e}^\ell$\\(OT-based multiplication)\end{center}
    };

   \node[anchor={north west}, box p1] (local5a)
    at ($(p1.south west)-(-.5,9)$) {
      $\share{\vec{e}^{\ell - 1}}_1$
    };
    \node[anchor={north east}, box p2] (local5b)
    at ($(p2.south east)-(.5,9)$) {
     $\share{\vec{e}^{\ell - 1}}_2$
    };

   \node[anchor={north west}, box p1] (local6a)
    at ($(p1.south west)-(-.5,10.5)$) {
      $\share{\mat{G}}_1$
    };
    \node[anchor={north east}, box p2] (local6b)
    at ($(p2.south east)-(.5,10.5)$) {
     $\share{\mat{G}}_2$
    };

    \draw[arrow] ($(mpc5.south west)+(.5, 0)$) -- ($(local5a.north east)+(0, -.1)$);
    \draw[arrow] ($(mpc5.south east)+(-.5, 0)$) -- ($(local5b.north west)+(0, -.1)$) ;
    \draw[arrow] ($(mpc6.south west)+(.5, 0)$) -- ($(local6a.north east)+(0, -.1)$);
    \draw[arrow] ($(mpc6.south east)+(-.5, 0)$) -- ($(local6b.north west)+(0, -.1)$) ;

     \draw[dashed, gray] ($(p1)+(2.7, -1.7)$) -- ($(p1)+(2.7, -11)$);
     \draw[dashed, gray] ($(p2)-(2.7, 1.7)$) -- ($(p2)-(2.7, 11)$);

  \end{tikzpicture}
  }
  \caption{Our protocol for the backward pass, corresponding to Algorithms~\ref{alg:fp_bp} from Section~\ref{section:ml}.
    \label{fig:bw-pass}
  }
\end{figure}

\section{Experiments}
\label{sec:experiments}
In this section, we present our experimental results for secure DNN training and prediction.

\myparagraph{Experimental Settings.} 
 The experiments were executed over two Microsoft Azure Ds32 v3 machines equipped with 128GB RAM and Intel Xeon E5-2673 v4 2.3GHz processor, running Ubuntu 16.04. In LAN experiments, machines were hosted in the same region (West Europe) with an average latency of 0.3ms and a bandwidth of 1.82GB/s.
 For WAN, the machines were hosted in two different regions (North Europe \& East US), with an average latency of 42ms and a bandwidth of 24.3 MB/sec. The machine specifications were chosen to be comparable with the ones in~\cite{mohassel2017secureml}, hence enabling direct running time comparisons.

\myparagraph{Implementation.} 
We use two distinct code bases. We use the EMP-toolkit~\cite{emp-toolkit} to implement our secure protocols for forward and backward passes, as described in Section~\ref{sec:protocols}. EMP is written in C++ and offers efficient implementations of OT and COT extension~\cite{asharov2013more}. We extended the semi-honest COT implementation to the functionality required for Protocol~\ref{prot:bool-bcomb-cot}, as it is currently limited to correlation functions of the form $f(x) = x \oplus\Delta$ (the ones required by Yao's garbled circuits protocol). This code base was used for timing results. 
We developed a more versatile insecure Python implementation based on Tensorflow\cite{abadi2016tensorflow} for accuracy experiments. While this implementation does not use MPC, it mirrors the functionality implemented using the EMP-toolkit. 

\myparagraph{Evaluations.} 
For training over QUOTIENT, we use two weight variables as described in  Section \ref{section:ml}: (i) ternary (2-bit) weights for the forward and backward passes, and (ii) 8-bit weights for the SGD and AMSgrad algorithms. We use 8-bits for the quantized weight gradients ($g$), activations ($a$) and activation gradients ($e$).

As our protocols are online, to 
compare with other approaches employing an offline phase, we take a conservative approach: we compare their offline computation time with our total computation time using similar computational resources. We adopt this strategy because online phases for these approaches are relatively inexpensive and could potentially involve a set-up overhead. Additionally, our model could easily be divided into offline/online phases, but we omit this for simplicity.

We employ a naive parallelization strategy: running independent processes over a mini-batch on different cores. This speeds up the computation on average by 8-15x over LAN and by about 10-100x over WAN depending on the number of parallelizable processes. We leave more involved parallelization strategies to future work.

\subsection{Data-independent benchmarking}
In this section, we present the running times of the basic building blocks that will be used for DNN prediction and training.
\begin{table}[]
\resizebox{\columnwidth}{!}{
\begin{tabular}{@{}lccccc@{}}
\toprule
Network              & $k$      & QUOTIENT (s)  & GC (s) & SecureML (OT) (s)& SecureML (LHE) (s) \\ \midrule
\multirow{5}{*}{LAN} & $10^{3}$ & 0.08                                                                        & 0.025  & 0.028  & 5.3     \\
                     & $10^{4}$ & 0.08                                                                        & 0.14   & 0.16   & 53      \\
                     & $10^{5}$ & 0.13                                                                        & 1.41   & 1.4    & 512     \\
                     & $10^{6}$ & 0.60                                                                        & 13.12  & 14*    & 5000*   \\
                     & $10^{7}$ & 6.0                                                                       & 139.80 & 140*   & 50000*  \\ \cmidrule(r){1-6}
\multirow{5}{*}{WAN} & $10^{3}$ & 1.7                                                                         & 1.9    & 1.4    & 6.2     \\
                     & $10^{4}$ & 1.7                                                                         & 3.7    & 12.5   & 62      \\
                     & $10^{5}$ & 2.6                                                                         & 20     & 140    & 641     \\
                     & $10^{6}$ & 7.3                                                                        & 148    & 1400*  & 6400*   \\
                     & $10^{7}$ & 44                                                                          & 1527   & 14000* & 64000*  \\ \cmidrule(l){1-6} 
 \end{tabular}
}
\caption{Comparison of our COT-based component-wise multiplication of $k$-dimensional vectors with ternary fixed-point multiplication using garbled circuits (GC) and SecureML~\cite{mohassel2017secureml} (OT, LHE). One of the vectors hold only ternary values.}
\label{tab:cp}
\end{table}

\begin{table}[t]
\resizebox{\columnwidth}{!}{
\begin{tabular}{@{}lcccc@{}}
\toprule
Network              & n    & QUOTIENT (s) & SecureML (OT Vec) (s) & \multicolumn{1}{l}{SecureML (LHE  Vec) (s)} \\ \midrule
\multirow{3}{*}{LAN} & 100  & 0.08         & 0.05       & 1.6                              \\
                     & 500  & 0.1          & 0.28       & 5.5                              \\
                     & 1000 & 0.14         & 0.46       & 10                               \\ \cmidrule(r){1-5}
\multirow{3}{*}{WAN} & 100  & 1.7          & 3.7        & 2                                \\
                     & 500  & 2            & 19         & 6.2                              \\
                     & 1000 & 2.7          & 34         & 11                               \\ \cmidrule(l){1-5} 
\end{tabular}
}
\caption{Performance comparison of our matrix-vector multiplication approach with the vectorized approaches of SecureML~\cite{mohassel2017secureml} (OT, LHE). Here we multiply a $128\times n$ ternary matrix with an $n$-dimensional vector.}
\label{tab:matprod}
\end{table}
\myparagraph{Component-wise Multiplication.}
Table~\ref{tab:cp} compares the running times of our COT-based approach from Protocol~\ref{prot:bool-bcomb-cot} for computing component-wise product with (i) an implementation of Algorithm~\ref{alg:ternary-mv} in a garbled circuit and (ii) two protocols proposed in SecureML~\cite{mohassel2017secureml} for the offline phase. Their first protocol corresponds to Gilboa's method for oblivious product evaluation (the OT-based variant implemented with a packing optimization). Their second is a sequence of Paillier encryptions, decryptions, and homomorphic additions (the LHE variant). 
For QUOTIENT and GC, one vector is ternary and the other holds $128$-bit values, while for OT and LHE both vectors hold $32$-bit values.
Although, our protocols for multiplication are independent and suitable for parallelization, here we benchmark without parallelization.
QUOTIENT clearly outperforms all other approaches as soon as we move past the set-up overhead of the base OTs. Note that most DNN layers involve greater than $10^4$ multiplications, which makes our approach more suitable for those applications.

\myparagraph{Matrix-Vector Product.}
As discussed in Section~\ref{sec:protocols}, our component-wise multiplication directly translates into matrix-vector products with local additions as described in Protocol~\ref{prot:ternary-mv}.
However, the protocols from~\cite{mohassel2017secureml} benefit greatly from a vectorization optimization, and thus a comparison of the matrix-vector multiplication tasks is important.
Table ~\ref{tab:matprod} compares the performance of implementation of Protocol~\ref{prot:ternary-mv}, for a ternary $128\times n$ matrix and an $n$-dimensional vector. Similar to Table \ref{tab:cp}, we populate the matrix with ternary values (2-bit values) and the vector with up to 128-bit values.
Our approach is at least 5x faster than the vectorized LHE protocol from \cite{mohassel2017secureml} on LAN, and is roughly 10x faster than the OT protocol from \cite{mohassel2017secureml} on WAN for $n\!\geq\!500$. In general, the speedup increases as we increase the number of computations.

\myparagraph{Layer Evaluation.}
Furthermore, we benchmark the basic building blocks of secure DNN training framework---forward and backward pass for a variety of different layer sizes. Figure~\ref{fig:fpstack} shows the running time of QUOTIENT for the forward pass as we increase the layer size. We split the total time into time spent on the matrix-vector product (Protocol \ref{prot:bool-bcomb}) and computation of activation function ($\relu$) using garbled circuits. Figure~\ref{fig:bpstack} shows the running time of the forward and backward pass, over a single batch of size $128$. We report the running time of each of the three required functionalities: quantization \& scaling, gradient computation, and error computation. As we increase the size of the layers, the garbled circuit for the quantization phase starts to dominate over the COT based matrix-matrix product required for the gradient computation. This can primarily be attributed to the quantization and scaling of the gradient matrix. Our COT-based matrix-vector multiplication shifts the bottleneck from the multiplication to the garbled circuits based scaling phase. Finding efficient protocols for that task is an interesting task for future work. In particular, parallelized garbled circuits and optimization of the matrix-matrix multiplication in the gradient computation phase could be explored.

\begin{figure}[t]
	\centering
	\begin{minipage}[t]{4cm}
		\centering
		\includegraphics[scale=0.3]{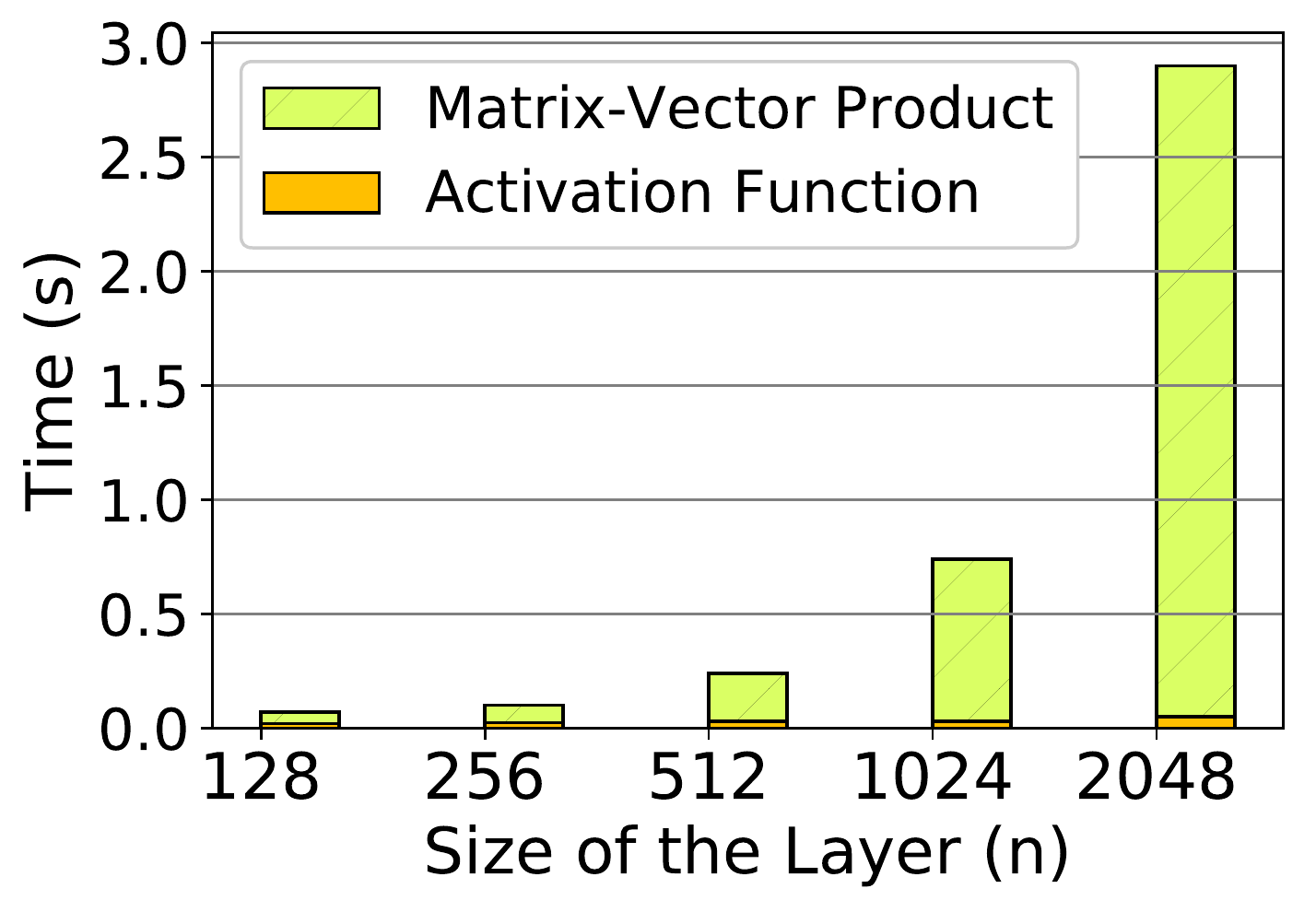}
		\caption{Forward pass time for single prediction over an $n \times n$ fully connected layer.}
		\label{fig:fpstack}
	\end{minipage}\hspace{0.2cm}
	\begin{minipage}[t]{4cm}
		\centering
		\includegraphics[scale=0.3]{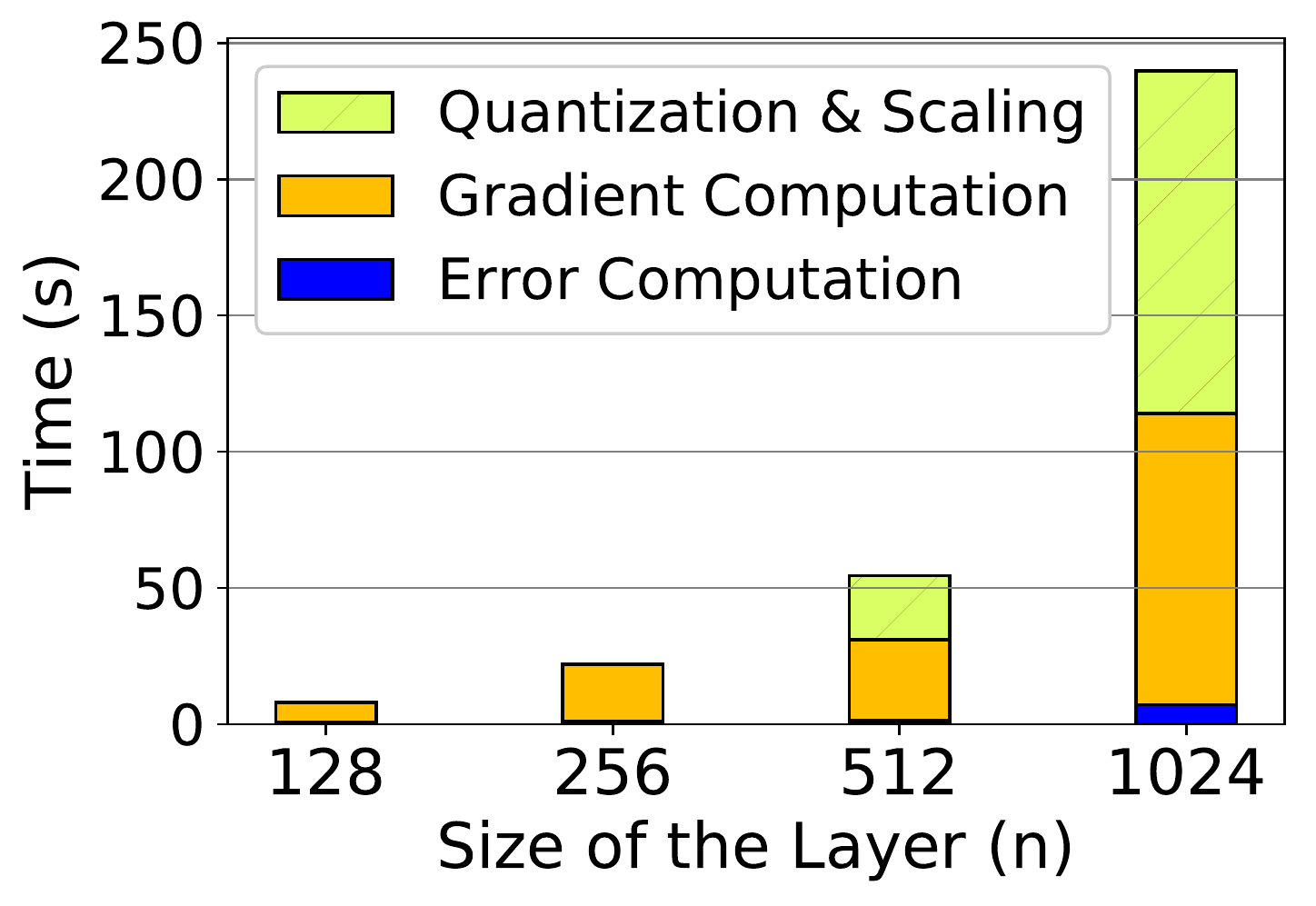}
		\caption{Forward and Backward pass time over an $n \times n$ fully-connected layer for 1 batch. Here batch size = 128.}
		\label{fig:bpstack}
	\end{minipage}
\end{figure}

\subsection{Experiments on Real-World Data}
In this section, we evaluate our proposed QUOTIENT 
on real-world datasets. 
We show that:  (i) QUOTIENT incurs only a small accuracy loss with respect to training over floating point, and (ii) it is more accurate than the state-of-the-art in fixed-point NN training, namely WAGE~\cite{wu2018training}. Both (i) and (ii) hold for several state-of-the-art architectures including fully connected, convolutional, and residual layers, and for different types of data (for residual layer results see the Appendix \ref{app:Residual}). For our 2PC protocols, we show that (iii) 2PC-QUOTIENT outperforms the existing secure training approach for DNNs SecureML \cite{mohassel2017secureml}, both in terms of accuracy and running time, and both in the LAN and WAN settings. We first report the accuracy across a variety of datasets to show (i) and (ii) above. Then we report running times for 2PC-QUOTIENT training and prediction to argue (iii).

\subsubsection{Datasets and Deep Neural Network Architectures.}
\label{sec:datasets}

We evaluate QUOTIENT on six different datasets. Five of these are privacy sensitive in nature, as they include health datasets---Thyroid~\cite{quinlan1987simplifying}, Breast cancer~\cite{cruz2014automatic}, MotionSense~\cite{malekzadeh2018mobile}, Skin cancer MNIST~\cite{DVN/DBW86T_2018} 
and a financial credit dataset---German credit~\cite{Dua2017}. We also evaluate our approach on MNIST~\cite{lecun1998gradient} for the purpose of benchmarking against prior work. 

\myparagraph{\textbf{MNIST}} contains 60K training and 10K test grayscale (28 $\times$ 28) images of 10 different handwritten digits. We adopt the state of the art floating point convolutional neural network LeNet~\cite{lecun1998gradient} (32C5-BN-32C5-BN-MP2-64C5-BN-64C5-BN-MP2-512FC-BN-DO-10SM)\footnote{BN, DO, and SM are batch-normalization, DropOut, and SoftMax respectively.} into a fixed-point equivalent of the form 32C5-MP2-64C5-MP2-512FC-10MSE for secure training \& inference. In addition, we explore a variety of fully-connected neural networks 2$\times$(128FC)-10MSE, 3$\times$(128FC)-10MSE, 2$\times$(512FC)-10MSE \& 3$\times$(512FC)-10MSE. We set the learning to $1$ for both SGD and AMSgrad optimizers.

\myparagraph{\textbf{MotionSense}} contains smartphone accelerometer and gyroscope sensor data for four distinct activities namely walking, jogging, upstairs, and downstairs. For each subject, the dataset contains about 30 minutes of continuously recorded data. We use a rolling window, of size 2.56 seconds each for extracting around 50K samples for training and 11K for testing. As proposed in ~\cite{malekzadeh2018mobile} we use floating point convolutional neural network 64C3-BN-MP2-DO-64C3-BN-MP2-DO-32C3-BN-MP2-DO-32C3-BN-MP2-DO-256FC-BN-DO-64FC-BN-DO-4SM and it's fixed-point analogue of the form 64C3-MP2-64C3-MP2-32C3-MP2-32C3-MP2-256FC-64FC-4MSE. We furthermore explore a fully connected architecture of the form 3$\times$(512FC)-4MSE.

\myparagraph{\textbf{Thyroid}} contains 3.7K training sample and 3.4K test samples of 21 dimensional patient data. The patients are grouped into three classes namely normal, hyperfunction and subnormal based on their thyroid functioning. We use a fully-connected neural network of the form 2$\times$(100FC)-3SM for this dataset and its analogue fixed-point network 2$\times$(100FC)-3MSE.

\myparagraph{\textbf{Breast cancer}} contains 5547 breast cancer histopathology RGB $(50 \times 50)$ images segregated into two classes---invasive ductal carcinoma and non-invasive ductal carcinoma. We use the 90:10 split for training and testing. We use a convolutional neural network with 3$\times$(36C3-MP2)-576FC-2SM and a fully-connected network of the form 3$\times$(512FC) along with their fixed point analogues.

\myparagraph{\textbf{Skin Cancer MNIST}} contains 8K training and 2K $(28 \times 28)$ dermatoscopic RGB images. They have been grouped into seven skin lesion categories. We use floating point network ResNet-20~\cite{he2016deep}. ResNet-20 is made up of batch-normalisation, dropout, SoftMax layers and employs cross-entropy loss for training in addition to the residual layers. For the fixed-point version of ResNet-20~\cite{he2016deep}, we exclude batch normalisation, average pooling, and SoftMax layers. In addition, we use a fully connected architecture of the form 2$\times$(512FC)-7MSE for secure training.

\myparagraph{\textbf{German credit}} contains 1k instances of bank account holders. they have been divided into two credit classes---Good or Bad. Each individual has 20 attributes (7 quantitative and 13 categorical). As a pre-processing step, we normalize the quantitative variables and encode the categorical variables using one-hot encoding. This amounts to 60 distinct feature for each individual in the dataset. We use the 80:20 split for training and testing. We use a fully-connected neural network of the form 2$\times$(124FC)-2SM and its fixed point analogue for this dataset.

\subsubsection{Accuracy.}
\label{sec:accuracy}

\begin{figure*}[t]
\centering
\includegraphics[width=\linewidth]{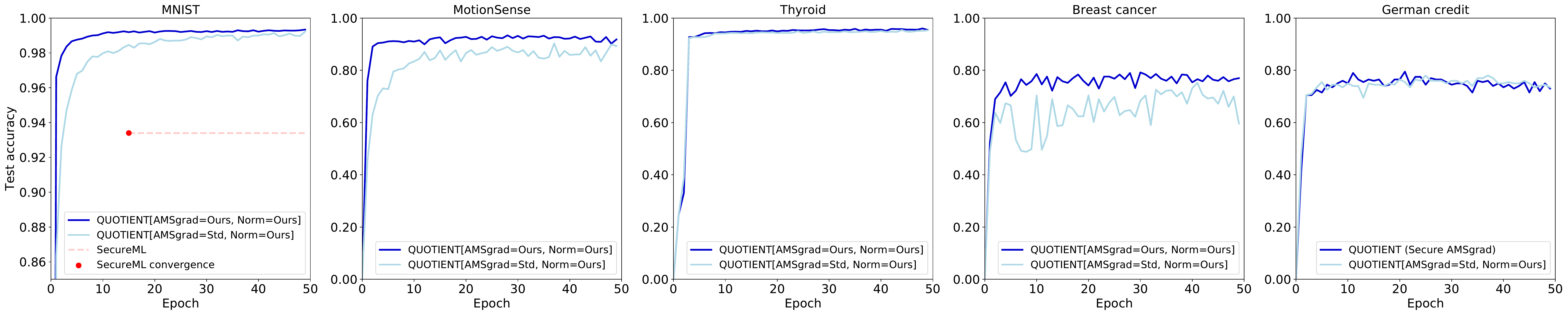}
\caption{Performance comparison of secure AMSgrad and secure SGD for QUOTIENT. The plots compare training curves over MNIST (using CNN), MotionSense, Thyroid, Breast cancer and German credit datasets.}
\label{fig:AccuracyAll}
\end{figure*}

\begin{figure*}[t]
\centering
\includegraphics[width=\linewidth]{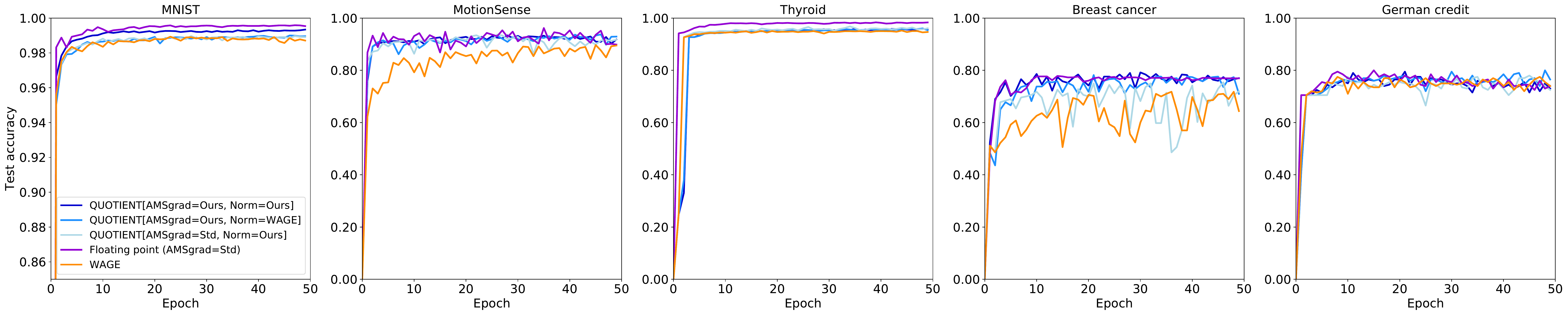}
\caption{Performance comparison of three different variants of QUOTIENT training with floating point and WAGE~\cite{wu2018training} training on MNIST, MotionSense, Thyroid, Breast cancer and German credit datasets.  {\QOW} and {\QSO}) differ from ({\QOO} in using next power of 2 vs the Closet Power of 2 and using standard AMSgrad by changing lines 9 and 11 of secure AMSgrad in Algorithm~\ref{alg:amsgrad}, respectively.}
\label{fig:MPCTricks}
\end{figure*}
\begin{table}[]

\resizebox{\columnwidth}{!}{
\begin{tabular}{@{}lcc@{}}
\toprule
Dataset       & \multicolumn{1}{l}{QUOTIENT  ($\%$)} & Floating point ($\%$) \\ \midrule
MNIST         & 99.38                                & 99.48                 \\
MotionSense   & 93.48                                & 95.65                 \\
Thyroid       & 97.03                                & 98.30                 \\
Breast cancer & 79.21                                & 80.00                 \\
German credit & 79.50                                & 80.50                 \\ \bottomrule
\end{tabular}
}
\caption{Accuracy comparison of training over state of the art floating point neural networks their fixed point equivalents using QUOTIENT.}
\label{tab:fpfp}
\end{table}
We evaluate the accuracy of QUOTIENT on different datasets and architectures. Also, we judge  
the impact of our MPC-friendly modifications
on accuracy. 
To do so we consider four variants of QUOTIENT: 
(i) Our proposed QUOTIENT, with secure AMSgrad optimizer (\QOO) (ii) QUOTIENT with the standard AMSgrad optimizer (\QSO) (described in Appendix~\ref{ref:appAMSgradStandard}) (iii) QUOTIENT with our proposed AMSgrad with the closest power of $2$ (C-Pow2) functionality instead of next power of $2$ for quantization (\QOW) (iv) QUOTIENT with the standard SGD optimizer (\QSSO).

\myparagraph{Comparison with Floating Point Training.} As a baseline evaluation of our training using \QOO, we compare its performance (upon convergence) with the floating point training counterparts in Table ~\ref{tab:fpfp}. For MNIST and Breast cancer datasets {\QOO} training achieves near state of the art accuracy levels for MNIST and Breast Cancer, while we differ by at most $\sim$2$\%$ for German credit, MotionSense and Thyroid datasets.

\myparagraph{Secure AMSgrad vs SGD.} Figure~\ref{fig:AccuracyAll} shows the training curves for QUOTIENT over all the datasets. In particular, we compare {\QOO} and {\QSSO}. The secure AMSgrad optimizer converges faster than the secure SGD optimizer especially for convolutional neural networks (used for MNIST, MotionSense, and Breast cancer). 

\myparagraph{Effects of our Optimizations.}
In order to evaluate the impact of our optimizations, we compare training with floating point and WAGE~\cite{wu2018training} with the first three variations of QUOTIENT
 in Figure~\ref{fig:MPCTricks}. For comparison with WAGE, we use the same fixed-point architecture as the ones used with QUOTIENT with an exception of the choice of optimizer--- we employ our proposed secure AMSgrad, while WAGE uses a SGD optimizer in conjunction with randomized rounding. Moreover, the floating point analogues employ (a) batch normalization and dropout in each layer, (b) softmax for the last layer, and (c) cross-entropy loss instead of MSE. We observe that {\QOO} outperforms all the other variants and is very close to the floating point training accuracy levels. We believe this is due to our modified AMSgrad, our normalization scheme, and use of the next power of two (used in {\QOO}) which may act as an additional regularization compared with the closest power of two.
\begin{table*}[t!]
\centering
\begin{adjustbox}{max width=\textwidth}
\begin{tabular}{l N N N N N N N N N N N N N N N N }
\hline
\multicolumn{17}{c }{LAN}\\
\hline
                                &\multicolumn{8}{c }{MNIST}& \multicolumn{2}{c }{MotionSense}& \multicolumn{2}{c }{Thyroid} & \multicolumn{2}{c }{Breast cancer}  & \multicolumn{2}{c}{German credit}\\ 
\cmidrule(lr){2-9}  
\cmidrule(lr){10-11}  
\cmidrule(lr){12-13}  
\cmidrule(lr){14-15} 
\cmidrule(lr){16-17} 

                                & \multicolumn{2}{c }{2 $\times$ (128FC)}  & \multicolumn{2}{c }{3 $\times$ (128FC)}     & \multicolumn{2}{c }{2 $\times$ (512FC)}     & \multicolumn{2}{c }{3 $\times$ (512FC)}    & \multicolumn{2}{c }{2 $\times$ (512FC)}   & \multicolumn{2}{c }{2 $\times$ (100FC)} & \multicolumn{2}{c }{3 $\times$ (512FC)}& \multicolumn{2}{c }{2 $\times$ (124FC)}    \\ 
\cmidrule(lr){2-3} 
\cmidrule(lr){4-5} 
\cmidrule(lr){6-7} 
\cmidrule(lr){8-9} 
\cmidrule(lr){10-11} 
\cmidrule(lr){12-13} 
\cmidrule(lr){14-15} 
\cmidrule(lr){16-17} 

 \multicolumn{1}{l}{Epoch}    & \multicolumn{1}{l }{Time} & \multicolumn{1}{c }{Acc} & \multicolumn{1}{c }{Time} & \multicolumn{1}{c }{Acc}  & \multicolumn{1}{c }{Time} & \multicolumn{1}{c }{Acc}& \multicolumn{1}{c }{Time} & \multicolumn{1}{c }{Acc}  & \multicolumn{1}{c }{Time} & \multicolumn{1}{c }{Acc} & \multicolumn{1}{c }{Time} & \multicolumn{1}{c }{Acc}& \multicolumn{1}{c }{Time}& \multicolumn{1}{c }{Acc}& \multicolumn{1}{c }{Time}& \multicolumn{1}{c }{Acc}\\ 

\cmidrule(lr){1-1}  
\cmidrule(lr){2-2} 
\cmidrule(lr){3-3} 
\cmidrule(lr){4-4} 
\cmidrule(lr){5-5} 
\cmidrule(lr){6-6}
\cmidrule(lr){7-7} 
\cmidrule(lr){8-8} 
\cmidrule(lr){9-9} 
\cmidrule(lr){10-10}
\cmidrule(lr){11-11} 
\cmidrule(lr){12-12} 
\cmidrule(lr){13-13} 
\cmidrule(lr){14-14} 
\cmidrule(lr){15-15} 
\cmidrule(lr){16-16} 
\cmidrule(lr){17-17}

 \multicolumn{1}{l}{1}    & \multicolumn{1}{l }{8.72h} & \multicolumn{1}{c }{0.8706} & \multicolumn{1}{c }{10.05h} & \multicolumn{1}{c }{0.9023}  & \multicolumn{1}{c }{27.13h} & \multicolumn{1}{c }{0.9341}& \multicolumn{1}{c }{38.76h} & \multicolumn{1}{c }{0.9424}  & \multicolumn{1}{c }{10.07h} & \multicolumn{1}{c }{0.8048}& \multicolumn{1}{c }{0.08h} & \multicolumn{1}{c }{0.2480} & \multicolumn{1}{c }{14.51h} & \multicolumn{1}{c }{0.4940} &  \multicolumn{1}{c }{0.03h} & \multicolumn{1}{c }{0.4100}\\ 
 
 \multicolumn{1}{l}{5}    & \multicolumn{1}{l }{43.60h} & \multicolumn{1}{c }{0.9438} & \multicolumn{1}{c }{50.25h} & \multicolumn{1}{c }{0.9536}  & \multicolumn{1}{c }{135.65h} & \multicolumn{1}{c }{0.9715}& \multicolumn{1}{c }{193.80h} & \multicolumn{1}{c }{0.9745}  & \multicolumn{1}{c }{50.35h} & \multicolumn{1}{c }{0.8847}& \multicolumn{1}{c }{0.40h} & \multicolumn{1}{c }{0.9341} & \multicolumn{1}{c }{72.55h} & \multicolumn{1}{c }{0.7360} &  \multicolumn{1}{c }{0.15h} & \multicolumn{1}{c }{0.725}\\ 
 
 \multicolumn{1}{l}{10}    & \multicolumn{1}{l }{87.20h} & \multicolumn{1}{c }{0.9504} & \multicolumn{1}{c }{100.50h} & \multicolumn{1}{c }{0.9604}  & \multicolumn{1}{c }{271.30h} & \multicolumn{1}{c }{0.9772}& \multicolumn{1}{c }{387.60h} & \multicolumn{1}{c }{0.9812}  & \multicolumn{1}{c }{100.70h} & \multicolumn{1}{c }{0.8855}& \multicolumn{1}{c }{0.80h} & \multicolumn{1}{c }{0.9453}& \multicolumn{1}{c }{145.10h} & \multicolumn{1}{c }{0.7600} &  \multicolumn{1}{c }{0.30h} & \multicolumn{1}{c }{0.7900}\\ 
 
 \hline
 \multicolumn{17}{c }{WAN}\\
 \hline
                                 &\multicolumn{8}{c }{MNIST}& \multicolumn{2}{c }{MotionSense}& \multicolumn{2}{c }{Thyroid} & \multicolumn{2}{c }{Breast cancer} &  \multicolumn{2}{c}{German credit}\\ 
\cmidrule(lr){2-9}  
\cmidrule(lr){10-11}  
\cmidrule(lr){12-13}  
\cmidrule(lr){14-15} 
\cmidrule(lr){16-17} 

                                & \multicolumn{2}{c }{2 $\times$ (128FC)}  & \multicolumn{2}{c }{3 $\times$ (128FC)}     & \multicolumn{2}{c }{2 $\times$ (512FC)}     & \multicolumn{2}{c }{3 $\times$ (512FC)}    & \multicolumn{2}{c }{2 $\times$ (512FC)}   & \multicolumn{2}{c }{2 $\times$ (100FC)} & \multicolumn{2}{c }{3 $\times$ (512FC)}&  \multicolumn{2}{c }{2 $\times$ (124FC)}    \\ 
\cmidrule(lr){2-3} 
\cmidrule(lr){4-5} 
\cmidrule(lr){6-7} 
\cmidrule(lr){8-9} 
\cmidrule(lr){10-11} 
\cmidrule(lr){12-13} 
\cmidrule(lr){14-15} 
\cmidrule(lr){16-17} 

 \multicolumn{1}{l}{Epoch}    & \multicolumn{1}{l }{Time} & \multicolumn{1}{c }{Acc} & \multicolumn{1}{c }{Time} & \multicolumn{1}{c }{Acc}  & \multicolumn{1}{c }{Time} & \multicolumn{1}{c }{Acc}& \multicolumn{1}{c }{Time} & \multicolumn{1}{c }{Acc}  & \multicolumn{1}{c }{Time} & \multicolumn{1}{c }{Acc} & \multicolumn{1}{c }{Time} & \multicolumn{1}{c }{Acc}& \multicolumn{1}{c }{Time}& \multicolumn{1}{c }{Acc}& \multicolumn{1}{c }{Time}& \multicolumn{1}{c }{Acc} \\ 

\cmidrule(lr){1-1}  
\cmidrule(lr){2-2} 
\cmidrule(lr){3-3} 
\cmidrule(lr){4-4} 
\cmidrule(lr){5-5} 
\cmidrule(lr){6-6}
\cmidrule(lr){7-7} 
\cmidrule(lr){8-8} 
\cmidrule(lr){9-9} 
\cmidrule(lr){10-10}
\cmidrule(lr){11-11} 
\cmidrule(lr){12-12} 
\cmidrule(lr){13-13} 
\cmidrule(lr){14-14} 
\cmidrule(lr){15-15} 
\cmidrule(lr){16-16} 
\cmidrule(lr){17-17}

 \multicolumn{1}{l}{1}    & \multicolumn{1}{l }{50.74h} & \multicolumn{1}{c }{0.8706} & \multicolumn{1}{c }{57.90h} & \multicolumn{1}{c }{0.9023}  & \multicolumn{1}{c }{139.71h} & \multicolumn{1}{c }{0.9341}& \multicolumn{1}{c }{190.10h} & \multicolumn{1}{c }{0.9424}  & \multicolumn{1}{c }{44.43h} & \multicolumn{1}{c }{0.8048}& \multicolumn{1}{c }{0.52h} & \multicolumn{1}{c }{0.2480} & \multicolumn{1}{c }{74.10h} & \multicolumn{1}{c }{0.4940} &  \multicolumn{1}{c }{0.15h} & \multicolumn{1}{c }{0.4100}\\ 
 
 \multicolumn{1}{l}{5}    & \multicolumn{1}{l }{253.7h} & \multicolumn{1}{c }{0.9438} & \multicolumn{1}{c }{289.5h} & \multicolumn{1}{c }{0.9536}  & \multicolumn{1}{c }{698.55h} & \multicolumn{1}{c }{0.9715}& \multicolumn{1}{c }{950.5h} & \multicolumn{1}{c }{0.9745}  & \multicolumn{1}{c }{222.15h} & \multicolumn{1}{c }{0.8847}& \multicolumn{1}{c }{2.6h} & \multicolumn{1}{c }{0.9341} & \multicolumn{1}{c }{370.5h} & \multicolumn{1}{c }{0.7360} &  \multicolumn{1}{c }{0.75h} & \multicolumn{1}{c }{0.725}\\ 
 
 \multicolumn{1}{l}{10}    & \multicolumn{1}{l }{507.4h} & \multicolumn{1}{c }{0.9504} & \multicolumn{1}{c }{579h} & \multicolumn{1}{c }{0.9604}  & \multicolumn{1}{c }{1397.1h} & \multicolumn{1}{c }{0.9772}& \multicolumn{1}{c }{1901h} & \multicolumn{1}{c }{0.9812}  & \multicolumn{1}{c }{444.3h} & \multicolumn{1}{c }{0.8855}& \multicolumn{1}{c }{5.2h} & \multicolumn{1}{c }{0.9453}& \multicolumn{1}{c }{741h} & \multicolumn{1}{c }{0.7600} & \multicolumn{1}{c }{1.5h} & \multicolumn{1}{c }{0.7900}\\ 
 \hline
\end{tabular}
\end{adjustbox}
\caption{Training time and accuracy values for various datasets and architectures after 1, 5 \& 10 training epochs using QUOTIENT over LAN and WAN.}
\label{tab:TrainingTimeLAN}
\end{table*}

\begin{table*}[t!]
\centering
\begin{adjustbox}{max width=\textwidth}
\begin{tabular}{l N N N N N N N N N N N }
\hline
                                &\multicolumn{5}{c }{MNIST}& \multicolumn{2}{c }{MotionSense} & \multicolumn{1}{c }{Thyroid} & \multicolumn{2}{c }{Breast cancer} & \multicolumn{1}{c}{German credit} \\ 
\cmidrule(lr){2-6}  
\cmidrule(lr){7-8}  
\cmidrule(lr){9-9}
\cmidrule(lr){10-11} 
\cmidrule(lr){12-12}
                                & \multicolumn{1}{c }{2 $\times$ (128FC)}  & \multicolumn{1}{c }{3 $\times$ (128FC)} & \multicolumn{1}{c }{2 $\times$ (512FC)} & \multicolumn{1}{c }{3 $\times$ (512FC)} & \multicolumn{1}{c }{Conv}     & \multicolumn{1}{c }{2 $\times$ (512FC)} & \multicolumn{1}{c }{Conv}     & \multicolumn{1}{c }{2 $\times$ (100FC)}   & \multicolumn{1}{c }{3 $\times$ (512FC)}      & \multicolumn{1}{c }{Conv}   & \multicolumn{1}{c }{2 $\times$ (124FC)}    \\ 
\cmidrule(lr){2-2} 
\cmidrule(lr){3-3} 
\cmidrule(lr){4-4} 
\cmidrule(lr){5-5} 
\cmidrule(lr){6-6}
\cmidrule(lr){7-7} 
\cmidrule(lr){8-8}
\cmidrule(lr){9-9}
\cmidrule(lr){10-10} 
\cmidrule(lr){11-11}
\cmidrule(lr){12-12}

 \multicolumn{1}{l}{Single Prediction(s)}    & \multicolumn{1}{c}{0.356} & \multicolumn{1}{c}{0.462} & \multicolumn{1}{c }{0.690} & \multicolumn{1}{c }{0.939} & \multicolumn{1}{c }{192} & \multicolumn{1}{c}{0.439}& \multicolumn{1}{c }{134} & \multicolumn{1}{c }{0.282}  & \multicolumn{1}{c }{3.58} & \multicolumn{1}{c }{62}  &  \multicolumn{1}{c }{0.272}\\ 

 \multicolumn{1}{l}{Batched Prediction (s)}    & \multicolumn{1}{c}{2.24} & \multicolumn{1}{c}{2.88} & \multicolumn{1}{c }{4.79} & \multicolumn{1}{c }{6.50} & \multicolumn{1}{c }{2226} & \multicolumn{1}{c }{2.91}& \multicolumn{1}{c }{1455} & \multicolumn{1}{c }{1.83}  & \multicolumn{1}{c }{44.02}& \multicolumn{1}{c }{447}  &  \multicolumn{1}{c }{1.77}\\ 
 
 \hline
\end{tabular}
\end{adjustbox}
\caption{Prediction time for various datasets and architectures using QUOTIENT over LAN. Here batch size = 128.}
\label{tab:PredictionTimeLAN}
\end{table*}

\begin{table*}[t]
\centering
\begin{adjustbox}{max width=\textwidth}
\begin{tabular}{l N N N N N N N N}
\hline
                                &\multicolumn{4}{c }{MNIST}& \multicolumn{1}{c }{MotionSense} & \multicolumn{1}{c }{Thyroid} & \multicolumn{1}{c }{Breast cancer} &  \multicolumn{1}{c}{German credit} \\ 
\cmidrule(lr){2-5}  
\cmidrule(lr){6-6}  
\cmidrule(lr){7-7}
\cmidrule(lr){8-8} 
\cmidrule(lr){9-9} 

                                & \multicolumn{1}{c }{2 $\times$ (128FC)}  & \multicolumn{1}{c }{3 $\times$ (128FC)} & \multicolumn{1}{c }{2 $\times$ (512FC)} & \multicolumn{1}{c }{3 $\times$ (512FC)}    & \multicolumn{1}{c }{2 $\times$ (512FC)} &      \multicolumn{1}{c }{2 $\times$ (100FC)}   & \multicolumn{1}{c }{3 $\times$ (512FC)}      &   \multicolumn{1}{c }{2 $\times$ (124FC)}    \\ 
\cmidrule(lr){2-2} 
\cmidrule(lr){3-3} 
\cmidrule(lr){4-4} 
\cmidrule(lr){5-5} 
\cmidrule(lr){6-6}
\cmidrule(lr){7-7} 
\cmidrule(lr){8-8}
\cmidrule(lr){9-9}

 \multicolumn{1}{l}{Single Prediction(s)}    & \multicolumn{1}{c}{6.8} & \multicolumn{1}{c}{8.8} & \multicolumn{1}{c }{14.4} & \multicolumn{1}{c }{19.9} &  \multicolumn{1}{c}{9.46}&  \multicolumn{1}{c }{5.99}  & \multicolumn{1}{c }{33.3} &    \multicolumn{1}{c }{5.1}\\ 

 \multicolumn{1}{l}{Batched Prediction (s)}    & \multicolumn{1}{c}{8.3} & \multicolumn{1}{c}{10.9} & \multicolumn{1}{c }{22.6} & \multicolumn{1}{c }{29.9} &  \multicolumn{1}{c }{12.08}&\multicolumn{1}{c }{6.89}  & \multicolumn{1}{c }{69.1}&    \multicolumn{1}{c }{7.3}\\ 
 
 \hline
\end{tabular}
\end{adjustbox}
\caption{Prediction time for various datasets and architectures using QUOTIENT over WAN. Here batch size = 128.}
\label{tab:PredictionTimeWAN}
\end{table*}

\subsubsection{End-to-End Running Times.} 
In the previous section, we demonstrated that large networks can match and improve upon the state-of-the-art in fixed-point deep networks. 
However, often one can use much simpler networks that are much faster and only sacrifice little accuracy. 
In order to balance accuracy and run-time, we design practical networks for each dataset and evaluate them here. 

\myparagraph{2PC-QUOTIENT Training.}
Table \ref{tab:TrainingTimeLAN} shows the running time of 2PC-QUOTIENT for practical networks across all datasets over LAN and WAN. 
We report accuracy and timings at 1,5 and 10 epochs. We note that the training time grows roughly linearly with the number of epochs. In all cases except the largest MNIST model, 10 epochs finish in under 12 days. Note that standard large deep models can easily take this long to train.

Our training protocols port nicely to the WAN settings. 
On  average, our networks are only about 5x slower over WAN than over LAN. This is is due to the low communication load and round-trip of our protocols. As a result, we present the first 2PC pragmatic neural network training solution over WAN.

\myparagraph{2PC-QUOTIENT Prediction.}
In addition to secure training, forward pass of QUOTIENT can be used for secure prediction. Table \ref{tab:PredictionTimeLAN} presents prediction timings for all datasets over LAN. In addition to fully connected networks, we also report the prediction timings over convolutional neural networks (we report residual networks in Appendix \ref{app:Residual}). The timings have been reported for a single point as well as $128$ parallel batched predictions (this corresponds to classifying many data points in one shot). Except for the multi-channel Breast cancer dataset, all single predictions take less than 1s and all batched predictions take less than 60s. 

Table \ref{tab:PredictionTimeWAN} presents an equivalent of Table \ref{tab:PredictionTimeLAN} under WAN settings. Here we limit ourselves to only fully connected architectures. While still practical, per prediction costs over WAN are, on an average, 20x slower than over LAN. This is due to the high initial setup overhead over WAN. However, batched predictions are only about 4-6x slower over WAN.

\myparagraph{Comparison with Previous Work.}
The only prior work for 2PC secure training of neural networks that we are aware of is SecureML \cite{mohassel2017secureml}. They report the results for  only fully connected neural network training on the MNIST dataset. 2PC-QUOTIENT is able to achieve its best accuracy levels in less than 16 hours over LAN and less than 90 hours over WAN. This amounts to a speedup of ~5x over LAN and a speedup of ~50x over WAN. In particular, QUOTIENT is able to make 2PC neural network training over WAN practical. Moreover, our 2PC-QUOTIENT training is able to achieve near state of the art accuracy of 99.38\% on MNIST dataset, amounting to an absolute improvement of 6\% over SecureML's error rates upon convergence. In terms of secure prediction, 2PC-QUOTIENT is 13x faster for single prediction and 7x faster for batched predictions over LAN in comparison to SecureML. Furthermore, 2PC-QUOTIENT offers a 3x speed-up for a single prediction and 50x for batched predictions versus SecureML over WAN.

\section{Conclusion}

In this paper, we introduced QUOTIENT, a new method to train deep neural network models securely that leverages oblivious transfer (OT). By simultaneously designing new machine learning techniques and new protocols tailored to machine learning training, QUOTIENT improves upon the state-of-the-art in both accuracy and speed. 
Our work is the first we are aware of to enable secure training of convolutional and residual layers, key building blocks of modern deep learning. However, training over convolutional networks is still slow, and incurs on a large communication load with our methods. Finding dedicated MPC protocols for fast evaluation of convolutional layers is an interesting venue for further work.

\begin{acks}
    Adri\`a Gasc\'on  and Matt J. Kusner were supported by The Alan Turing Institute under the \grantsponsor{epsrc}{EPSRC}{https://epsrc.ukri.org/} grant \grantnum{epsrc}{EP/N510129/1}, and funding from the UK Government's Defence \& Security Programme in support of the Alan Turing Institute.
\end{acks}


\appendix

\section{Standard AMSgrad Optimizer}

\label{ref:appAMSgradStandard}
Algorithm~\ref{alg:standardamsgrad} describes the steps for training deep neural networks using standard AMSgrad optimizer. This can be summarised as: (i) Sampling the input pair $(\ab^0, y)$ from the dataset $ \mathcal{D}$; (ii) Using the input $(\ab^0)$ to obtain the prediction $\{\ab^l\}_{l=1}^L$; (iii) Computing the loss between the prediction $\{\ab^l\}_{l=1}^L$ and the target output $(y)$, and computing the gradient $\Gb^l$ for the loss $\ell(\ab^L_i, y_i)$ with respect to each of the weights $\Wb_l$ in the network; (iv) Updating the weights using a weighted average of the past gradients, specifically, their first and second moments $\Mb^l$ and $\Vb^l$. 

\begin{algorithm2e}[t]
\DontPrintSemicolon
\SetKwComment{Comment}{{\scriptsize$\triangleright$\ }}{}
\SetKwInOut{State}{State}
\caption{Standard AMSgrad Optim.\ ($\Delta_{\textrm{ams}}$)}\label{alg:standardamsgrad}
		\KwIn{Weights $\{\Wb^l\}_{l=1}^L$, \\ Data $\mathcal{D}=\{(\ab^0_{i},y_{i}\}_{i=1}^n$, learning rate $\eta$}
        \KwOut{Updated weights $\{\Wb^l\}_{l=1}^L$}
\BlankLine
\begin{minipage}{\hsize}
  \begin{algorithmic}[1]
  	\STATE Initialize: $\{\Mb^{l}, \Vb^{l}\}_{l=1}^L =0$
  	\FOR{$t = 1,\ldots,T$}
  		\STATE $(\ab^0_{(p_a)}, y_{(p_a)}) \sim \mathcal{D}$
  		\STATE $\{\ab^l\}_{l=1}^L =$ \texttt{Forward}$(\{\Wb	^l\}_{l=1}^L, \ab^0)$ 
  		\STATE $\{\Gb^l\}_{l=1}^L =$ \texttt{Backward}$(\{\Wb^l, \ab^l\}_{l=1}^L, y)$
		\STATE $\{\Mb^l = 0.9 \Mb^l + 0.1 \Gb^l \}_{l=1}^L$ \Comment*[r]{{\scriptsize weighted mean}}
		\STATE $\{\Vb^l = 0.99 \Vb^l + 0.01 \big(\Gb^l\big)^2 \}_{l=1}^L$ \Comment*[r]{{\scriptsize weighted variance }}
		\STATE $\{ \hat{\Vb}^l = \max(\hat{\Vb}^l, \Vb^l) \}_{l=1}^L$
		\STATE $\{\Gb^l = \frac{\Mb^l}{ \sqrt{\hat{\Vb}^l} + \epsilon  } \}_{l=1}^L$ \Comment*[r]{{\scriptsize history-based scaling with square root}}
  		\STATE $\{\Wb^l = \Wb^l - \eta \Gb^l\}_{l=1}^L$ 
  	\ENDFOR
\end{algorithmic}
\end{minipage}
\end{algorithm2e}

\section{Experiments on Residual Layers}
\label{app:Residual}
In addition to fully-connected and convolutional layers, we also evaluate 2PC-QUOTIENT on residual neural networks. Table~ \ref{tab:PredictionTimeLANSkin} shows the prediction time of 2PC-QUOTIENT on Skin cancer MNIST dataset using for both fully-connected and residual neural networks, while we show its training timing using practical fully-connected neural networks in Table. \ref{tab:TrainingTimeLANSkin}.
\begin{table}[t]
\centering
\begin{tabular}{l N N N }
\hline
                                & \multicolumn{2}{c }{LAN} &  \multicolumn{1}{c }{WAN}  \\ 
\cmidrule(lr){2-3}
\cmidrule(lr){4-4} 

                                & \multicolumn{1}{c }{2 $\times$ (512FC)}& \multicolumn{1}{c }{ResNet} &  \multicolumn{1}{c }{2 $\times$ (512FC)} \\ 
\cmidrule(lr){2-2} 
\cmidrule(lr){3-3} 
\cmidrule(lr){4-4}
 \multicolumn{1}{l}{Single Prediction(s)}    &  \multicolumn{1}{c }{1.269}& \multicolumn{1}{c }{1982} & \multicolumn{1}{c }{17} \\ 

 \multicolumn{1}{l}{Batched Prediction (s)}    &  \multicolumn{1}{c }{14.31}& \multicolumn{1}{c }{18122} & \multicolumn{1}{c }{39.9} \\ 
 
 \hline
\end{tabular}
\caption{Prediction time for Residual and fully-connected neural networks on Skin cancer MNIST using QUOTIENT over LAN. Here batch size = 128.}
\label{tab:PredictionTimeLANSkin}
\end{table}

\begin{table}[t]
\centering
\begin{tabular}{l N N N N }
\hline
                             & \multicolumn{2}{c }{2 $\times$ (512FC)}  \\ 
\cmidrule(lr){1-1} 
\cmidrule(lr){2-4} 
 \multicolumn{1}{l}{Epoch}    &  \multicolumn{1}{c }{LAN} & \multicolumn{1}{l }{WAN} & \multicolumn{1}{l }{Acc}\\ 
\cmidrule(lr){1-1}  
\cmidrule(lr){2-2}  
\cmidrule(lr){3-3} 
\cmidrule(lr){4-4}
 \multicolumn{1}{l}{1}    &  \multicolumn{1}{c }{8.57h}  &  \multicolumn{1}{c }{38.44h} &  \multicolumn{1}{c }{0.2078}\\
 \multicolumn{1}{l}{5}     & \multicolumn{1}{c }{42.85h} &  \multicolumn{1}{c }{192.2h}  &  \multicolumn{1}{c }{0.7026}\\ 
 \multicolumn{1}{l}{10}    & \multicolumn{1}{c }{85.70h}  &  \multicolumn{1}{c }{384.4h}  &  \multicolumn{1}{c }{0.7157}\\ 
 \hline
\end{tabular}
\caption{Training time and accuracy values of QUOTIENT on Skin cancer MNSIT dataset after 1, 5 \& 10 training epochs using fully-connected neural network over both LAN and WAN.}
\label{tab:TrainingTimeLANSkin}
\end{table}

\section{Proof of Protocol~\ref{prot:bool-bcomb}}
\label{sec:appendix-proof}

\begin{lemma}
Let $\vec{b}$ and $\vec{a}$ be a Boolean and integer vector, respectively, shared among parties $\A, \B$.
Given a secure OT protocol, the two-party Protocol~\ref{prot:bool-bcomb} is secure against semi-honest adversaries, and computes an additive share of the inner product $\vec{b}^\top \vec{a}$ among $\A, \B$.
\end{lemma}
\begin{proof}
For the correctness, note that, for all $j\in [m]$, $b_j a_j = b_j\share{a_j}_1 + b_j\share{a_j}_2$, and that each of the OTs in steps $4$ and $5$ of the algorithm are used to compute an additive share of each of the terms of the sum. Concretely, in the OT in line $4$ is used to compute an additive share of $b_j\share{a_j}_1$ as $z_{1,j}+z_{1,j}'$, and the value received by party $\B$ is $m_{1,\share{b_j}_2} = -z_{1,j}$ if $\share{b_j}_1 \oplus \share{b_j}_2 = 0$
and $m_{1,\share{b_j}_2} = \share{a_j}_i-z_{1,j}$ otherwise.
Security follows easily from the fact that all messages in the OTs are masked with fresh randomness $z_{i,j}$ and thus constructing simulators from a simulator for OT is straightforward.
\end{proof}

\section{Proof of Protocol \ref{prot:bool-bcomb-cot}}
\label{app:proofbool-bcomb-cot}

\begin{proof}
Privacy follows directly from the correctness of the COT subprotocol. To see that the protocol computes the right value, we argue 
for a $j\in [m]$, and distinguish cases according to all possible values of the shares of $w_j$.

\noindent
\textbf{Case I} ($w_j=0$; $\share{w_j}_1=1$, $\share{w_j}_2=1$): In this case $f^1(x)=x - \share{a_j}_1$ and $f^2(x') = x' - \share{a_j}_2$.
Upon execution of step 2,  $P_1$ obtains $x$ and $P_2$ obtains $y = x - \share{a_j}_1$.
Upon execution of step 3, $P_2$ obtains $x'$ and $P_1$ obtains $y' = x' - \share{a_j}_2$.
$P_1$ accumulates  $\share{z_j}_1 = \share{a_j}_1 - x + x' - \share{a_j}_2$ into $\share{z}_1$ and $P_2$ accumulates $\share{z_j}_2 = \share{a_j}_1 - x' + x - \share{a_j}_1$ into $\share{z}_2$.
Then $\share{z_j}_1 + \share{z_j}_2 = 0$, which is $w_ja_j$ as $w_j=0$.

\noindent
\textbf{Case II} ($w_j=0$; $\share{w_j}_1=0$, $\share{w_j}_2=0$):
In this case $f^1(x)=x + \share{a_j}_1$ and $f^2(x') = x' + \share{a_j}_2$.
Upon execution of step 2, $P_1$ obtains $x$ and $P_2$ obtains $y = x$.
Upon execution of step 3, $P_2$ obtains $x'$ and $P_1$ obtains $y' = x'$
$P_1$ accumulates  $\share{z_j}_1 = - x + x'$ into $\share{z}_1$ and $P_2$ accumulates $\share{z_j}_2 = - x' + x$ into $\share{z}_2$.
Then $\share{z_j}_1 + \share{z_j}_2 = 0$, which is  $w_ja_j$ as $w_j=0$.

\noindent
\textbf{Case III} ($w_j=1$; $\share{w_j}_1=0$, $\share{w_j}_2=1$): 
In this case $f^1(x) = x + \share{a_j}_1$ and $f^2(x') = x' - \share{a_j}_2$.
Upon execution of step 2, $P_1$ obtains $x$ and $P_2$ obtains $y = x + \share{a_j}_1$.
Upon execution of step 3, $P_2$ obtains $x'$, $P_1$ obtains $y' = x'$.
$P_1$ accumulates  $\share{z_j}_1 = - x + x'$ into $\share{z}_1$ and $P_2$ accumulates $\share{z_j}_2 = \share{a_j}_2 - x' + x + \share{a_j}_1$ into $\share{z}_2$.
Then $\share{z_j}_1 + \share{z_j}_2 = a_j$, which is  $w_ja_j$ as $w_j=1$.

\noindent
\textbf{Case IV}  ($w_j=1$; $\share{w_j}_1=1$, $\share{w_j}_2=0$):
In this case $f^1(x)=x - \share{a_j}_1$ and $f^2(x') = x' + \share{a_j}_2$.
Upon execution of step 2, $P_1$ obtains $x$ and  $P_2$ obtains $y = x$.
Upon execution of step 3, $P_2$ obtains $x'$ and $P_1$ obtains $y' = x' + \share{a_j}_2$.
$P_1$ accumulates  $\share{z_j}_1 = \share{a_j}_1 - x + x' + \share{a_j}_2$ into $\share{z}_1$, $P_2$ accumulates $\share{z_j}_2 = - x' + x $ into $\share{z}_2$.
Then, $\share{z_j}_1 + \share{z_j}_2 = a_j$, which is $w_ja_j$ as $w_j=1$.
\end{proof}

\end{document}